\documentclass[journal]{IEEEtran}

\IEEEoverridecommandlockouts
\usepackage{blindtext, graphicx}
\usepackage{cite}
\usepackage{graphicx}
\usepackage{amsmath}
\usepackage{amsfonts}
\usepackage{color}
\usepackage{enumerate}
\usepackage{subcaption}
\usepackage[export]{adjustbox}
\usepackage{caption}
\usepackage{amssymb}
\usepackage{url}
\usepackage{color,soul}
\usepackage[english]{babel}
\usepackage{hyperref}
\usepackage{multirow}
\usepackage{multicol}
\newtheorem{theorem}{Theorem}

\newtheorem{definition}{Definition}
\newtheorem{assumption}{Assumption}

\newtheorem{lemma}{Lemma}

\renewcommand{\vec}[1]{\mathbf{#1}}

\newtheorem{remark}{Remark}

\makeatletter
\newcommand*{\rom}[1]{\expandafter\@slowromancap\romannumeral #1@}
\makeatother

\newenvironment{proof}{\textit{Proof:}}

\usepackage{algorithmic}

\begin{document}

\title{Formation Control of Nonlinear Multi-Agent \\ Systems Using Three-Layer Neural Networks}
\author{Kiarash Aryankia$^{ \href{https://orcid.org/0000-0002-4751-3925}{\includegraphics[scale=.04]{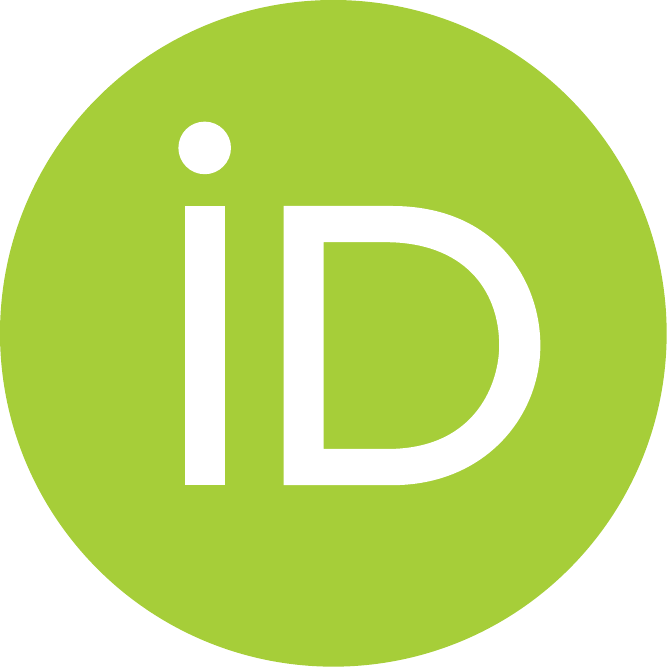}}}$, \IEEEmembership{Student Member, IEEE} and Rastko R. Selmic$^{\href{https://orcid.org/0000-0001-9345-8077}{\includegraphics[scale=.04]{Pictures/ORCID.pdf}}}$, \IEEEmembership{Senior Member, IEEE}
\thanks{This work was supported by the Natural Sciences and Engineering Research Council of Canada (NSERC) Discovery Grant \#RGPIN-2018-05093. (\textit{Corresponding Author: Kiarash Aryankia}). The authors are with the Department of Electrical and Computer Engineering,
    Concordia University, Montreal, QC, Canada
    {\tt\small (email: k\_aryank@encs.concordia.ca; \tt\small rastko.selmic@concordia.ca).}}}

\maketitle
\IEEEpeerreviewmaketitle

\begin{abstract}
This paper considers a leader-following formation control problem for heterogeneous, second-order, uncertain, input-affine, nonlinear multi-agent systems modeled by a directed graph. A tunable, three-layer neural network (NN) is proposed with an input layer, two hidden layers, and an output layer to approximate an unknown nonlinearity.  Unlike commonly used trial and error efforts to select the number of neurons in a conventional NN, in this case an \textit{a priori} knowledge allows one to set up the number of neurons in each layer. The NN weights tuning laws are derived using the Lyapunov theory. The leader-following and formation control problems are addressed by a robust integral of the sign of the error (RISE) feedback and a NN-based control. The RISE feedback term compensates for unknown leader dynamics and the unknown, bounded disturbance in the agent error dynamics. The NN-based term compensates for the unknown nonlinearity in the dynamics of multi-agent systems, and semi-global asymptotic tracking results are rigorously proven using the Lyapunov stability theory. The results of the paper are compared with two previous results to evaluate the efficiency and performance of the proposed method.
\end{abstract}

\begin{IEEEkeywords}
  Formation control, Multi-layer neural networks, Nonlinear multi-agent systems, Second-order systems.
\end{IEEEkeywords}

\section{Introduction}

\IEEEPARstart{M}{utli-agent} systems have received significant attention from researchers and have gained more attraction due to their various applications. One of the significant applications of multi-agent systems is the formation control which has been inspired by nature. In \cite{formation}, different categories of formation control problems were reviewed. The displacement control, one of the main categories of formation control problems, includes leader-following consensus control problem (also known as distributed cooperative tracking control problem) as a particular case that was widely studied in the literature, e.g., \cite{IET,lewis,displacement,das2010distributed}. 
In the formation control problem of multi-agent systems, it is common to consider single- and double-integrator dynamics \cite{zhang2016finite,displacement,ren2007multi} or nonlinear dynamics \cite{lewis,IET,chen2014cooperative,shi2019adaptive} to describe the dynamics of multi-agent systems. Most of the literature that considers an unknown nonlinear dynamics, utilizes NNs or fuzzy logic controllers to approximate the uncertain nonlinearity in the system dynamics. There are two main reasons for this: first, the multi-layer NNs and fuzzy logic controllers are universal approximators that can approximate any continuous and smooth function over a compact set at a desired accuracy \cite{castro1995fuzzy,lewis}; and second, a precise mathematical model for the dynamics of a multi-agent system is not required, which eliminates an extra effort to acquire an accurate system dynamics \cite{park2005direct}. 

 The literature on formation control mainly uses one-layer NNs, e.g., radial basis function neural networks (RBFNNs) \cite{tan2020neural,IET} or two-layer NNs \cite{chen2014cooperative}, to compensate for the unknown nonlinearity (usually caused by uncertainty) of system dynamics. These studies commonly use the universal approximation property where the RBFNNs and the Gaussian activation functions are utilized. In RBFNN, hidden layer neurons are usually uniformly distributed on a regular lattice or they are randomly initialized. In \cite{wang2006learning}, authors constructed an RBFNN where a minimum distance of any two centers of the bell-shaped activation functions is the same as the width of the activation function.
 An RBFNN can approximate a nonlinearity if the input variables of the activation function are located in certain neighborhoods of the RBF networks centers \cite{wang2006learning}. 
 Thus, one can select a large number of neurons (not a priori known) and distribute activation functions over a regular lattice uniformly to cover the compact set. A one-layer NN with a tunable, hidden layer, in the form of linear-in-the-parameter can approximate a nonlinearity in the dynamics of a nonlinear system at a desired accuracy. Still, the exact number of neurons for control systems design is not clearly established. Although the self-organizing receptive field method was introduced in \cite{moody1988learning} to optimize the receptive fields locations in RBFNN, these results have not been used in the formation control problem to approximate the unknown nonlinearity. 

These limitations and restrictions motivated us to propose the three-layer NN with two hidden layers, where the numbers of hidden neurons are clearly given (see Remark \ref{remark_6}). Here we use the sigmoid activation function.

To address the number of neurons in a NN, a self-structuring NN was proposed in \cite{park2005direct}, whereby a neuron divides into two neurons by satisfying a specific condition, but the output of these two neurons remains identical to the output of a neuron without division. This work was extended to a nonlinear multi-agent system \cite{chen2014cooperative}, where authors designed an observer to estimate states of leader for each agent and addressed cooperative tracking control problem of a leader. Compared with existing results, we use two hidden layers instead of one \cite{park2005direct,chen2014cooperative,patre2008asymptotic,shin2018neural,yang2015robust}, and we developed tuning laws for all three layers. The other novelty is that the number of neurons is fixed. 

To avoid the chattering effect of the sliding mode controllers, the RISE feedback controller has gained significant attention. Although a RISE feedback control can potentially exceed the actuators limit due to the large integral gain, it has been shown that it can be implementable in real-time applications, e.g., \cite{yao2018adaptive}. The RISE feedback has been used in \cite{xian2004continuous}, to compensate for a nonlinearity in a class of higher-order, multi-input multi-output, nonlinear systems. These results were extended in \cite{patre2008asymptotic}, where authors used the NN-based controller to address the tracking problem of an uncertain system. In \cite{patre2008asymptotic}, the authors used the discontinuous projection algorithm, which requires an a priori knowledge of the convex set and maximum and minimum of NN's weights \cite{yao1997adaptive}. In a similar study, \cite{dierks2008neural}, authors used two-layer NN (input, hidden, and output layers) where only the weights between two outer layers can be tuned. The tuning laws in \cite{dierks2008neural} are derived from Lyapunov stability theory. These results further extended in \cite{shin2018neural,yao2018adaptive}, where the authors used a two-layer NN, and a discontinuous projection algorithm has tuned the NN weights matrices between any two consecutive layers. 

In comparison to \cite{xian2004continuous,patre2008asymptotic,yao1997adaptive,dierks2008neural,shin2018neural,yao2018adaptive}, the novelty of our work is that we utilize a three-layer NN comprised of an input layer, two hidden layers, and an output layer. The NN tuning laws are derived using the Lyapunov theory. Moreover, all NN weights matrices are tunable, and the tuning law for each of them has been established.

Several of the NN-based formation control results that address the displacement-based, leader-following problem established only uniform ultimate bounded stability, due to NN approximation error and lack of communication from leader to all agents, e.g., \cite{zhang2012adaptive,chen2014cooperative,chen,chendelay,lewis,das2010distributed}. We propose to use a RISE feedback controller with two hidden layers NN which achieves semi-global, asymptotic stability for the leader-following formation control for heterogeneous, second-order, uncertain, input-affine, nonlinear multi-agent systems.

The main contributions of the paper are as follows:
\begin{enumerate}
    \item We consider the leader to be connected to at least one of the other agents. Unlike \cite{yang2015robust,patre2008asymptotic,fischer2013saturated,dierks2008neural,shin2018neural,yao2018adaptive}, where the desired trajectory signals (information from the leader) are fed into the first layer of the NN, the novelty here is to construct the NN without these signals. Thus, effects of the leader dynamics and unknown disturbances in the tracking error of each agent are ameliorated by the RISE feedback controller.
    
    \item We determine the number of neurons in each layer and derived the tuning laws for the NN weights matrices. Although the literature considered one- or two-layer NN \cite{IET,das2010distributed,yang2015robust,ge2004adaptive,chendelay,shin2018neural,yao2018adaptive}, we developed our results for a three-layer NN. 
    
    \item The semi-global, asymptotic, leader-following performance is achieved, which distinguishes our work from \cite{lewis,chen,IET,das2010distributed}.

\end{enumerate}

This paper is organized as follows. In Section II, the basic graph theory, problem statement, and NN design with two hidden layers are given. In Section III, the control design of RISE feedback and NN adaptive tuning laws for a second-order, multi-agent system with disturbance and unknown nonlinearity in the system dynamics are given, where the stability of the proposed method is proven using the Lyapunov stability theory. In Section IV, the numerical simulation results are presented to substantiate the theoretical results. Then, results are compared with two different methods to illustrate the performance of the proposed method. Finally, the conclusion is given in Section V.

\section{Problem Formulation and Preliminaries}

\subsection{Notations}
The notations that are used in the paper are given here. Let ${\mathbb{R}^n}$ indicate $n$-dimensional Euclidean space and $\mathbb{R}^{n\times m}$ denote the set of $n\times m$ real matrices. Let ${I}_n$ denote $n\times n$ identity matrix and ${\mathbf{1}}_N=[1,...,1]^T$ is an all-one vector of size $N$. With $diag(a_1,a_2,...)$, we denote a block diagonal matrix with diagonal elements of $a_i$. The notation vec(.) represents a stacked vector of a matrix column-wise. For a square matrix $A$, $tr(A)$ is the sum of diagonal elements of matrix $A$. The Kronecker product is represented by $\otimes$. With $A^T$, we indicate the transpose of matrix A. For a matrix $A$, we denote the minimum and maximum singular values by $\underline{\sigma}(A)$ and $\bar{\sigma}(A)$, respectively. By $||.||$ we denote the 2-norm  of a vector or matrix and the Frobenius norm for an arbitrary matrix $A$ is defined as $||A||_F=  \sqrt{tr(A^TA)}$. For any vector $x \in \mathbb{R}^n$, the Manhattan norm is $||x||_1 =\sum_{i=1}^n |x_i|$.  With a superscript $'$ we denote the partial derivative of a function with respect to the given variable. The signum function is denoted by sgn(.). The notation $\underset{\mu (H)=0}{\cap}$ denotes the intersection over all sets $H$ of Lebesgue measure zero, and $\overline{co}$ denotes the convex closure.

\subsection{Graph Theory}
We model interactions among agents by a directed graph $G=(V,E)$, where $V=\{v_1, ... , v_N\}$ is a non-empty set of $N$ vertices and $ E\subset V \times V $ is a set of edges. We assume that the graph is simple, i.e., no repeated edges and no self-loops.
Adjacency, or connectivity matrix is $A=[a_{ij}]$ with $a_{ij}>0$, if $(v_j,v_i) \in E$, and $a_{ij}=0$, otherwise. Note that we consider $a_{ii}=0$. 
\begin{definition}
        The neighborhood $N_i$ of the vertex $v_i$ is a set of all the edges adjacent to $v_i$
        \begin{equation}\label{eq15_3}
            \begin{array}{lr}
                N_i=\{V_j \in V |\ (v_j, v_i) \in E\}.
            \end{array}
        \end{equation}
    \end{definition}

Let us define the in-degree matrix $D=diag(deg_i)$, with ${deg}_i=\sum_{j\in N_i} a_{ij} $ as the sum of incoming edges weights to the vertex $i$. The graph Laplacian is $L=D-A$, where a sum of all rows is equal to zero. 

The leader adjacency matrix is defined as $B=diag(b_1,...b_N)$, with $b_i>0$, if the leader sends its information to the $i$-th agent, and $b_i=0$, otherwise. It is assumed that the leader is connected to at least one agent.  \begin{lemma}[\hspace{-.1mm}\cite{wah2007synchronization}]
     The directed graph $G$ is strongly connected if and only if its Laplacian matrix is irreducible.
\end{lemma}
\begin{definition}[\hspace{-.1mm}\cite{qu2009cooperative,lewis}]\label{def_m-matrix}
      A square matrix $P \in \mathbb{R}^{N\times N}$ is said to be a non-singular M-matrix, if it is positive definite, and the off-diagonal elements have non-positive values.
\end{definition}

\begin{lemma}[\hspace{-.1mm}\cite{qu2009cooperative}]\label{lem_l+b}
Let the directed graph $G=(V,E)$ be strongly connected. If $\exists  b_i\ne 0$, then $L+B$ is a non-singular M-matrix and the following inequality holds
\begin{equation}\label{eq_2}
    \Pi \triangleq P(L+B) + (L+B)^T P >0,
\end{equation}
where $P=diag(p_i) \triangleq diag(1/q_i)$ is a positive definite matrix and $\vec{q} \triangleq (L+B)^{-1} {\mathbf{1}}_N$ with $\vec{q}= [q_1,...,q_N]^T$.
\end{lemma}

\subsection{Problem Formulation}
The multi-agent system is considered to consist of $N$ agents where the dynamics of each agent is given by
\begin{equation}\label{eq1}
\begin{array}{lr} 
\dot{p}_{i}=v_{i}, \\
\dot{v}_{i}=f_{i}(p_{i},v_i)+g_i(p_{i},v_i)u_{i}(t) +w_{i}(t),   \ \ i=1,\ ... \ ,N,
\end{array}
\end{equation}
\noindent where the variables $p_{i}\in\mathbb{R}^n$ and $v_{i}\in \mathbb{R}^n $ denote the position and velocity states of each agent, respectively. The functions  $f_i(p_i,v_i)\in \mathbb{R}^{n}$ and $g_i(p_i,v_i) \in \mathbb{R}^{n\times n}$ are unknown nonlinear $\mathcal{C}^2$ functions. The control input is $u_i\in \mathbb{R}^n$ and a disturbance of an each agent is $w_i \in \mathbb{R}^n$. 

One can rewrite the dynamics of the multi-agent system in a stacked form as follows:
    \begin{equation}\label{eq2_1}
    \begin{split}
    \dot{\vec{x}}_1&=\vec{x}_{2},  \\
    \dot{\vec{x}}_2&=f(\vec{x_1},\vec{x_2})+g(\vec{x_1},\vec{x_2})u +w,
    \end{split}
    \end{equation}
where the stacked vector $\vec{x}_1=[{p^T_1}, {p^T_2}, ...,  {p^T_N}]^T\in \mathbb{R}^{nN}$ and the stacked vector $\vec{x}_2=[{v^T_1}, {v^T_2}, ..., {v^T_N}]^T \in \mathbb{R}^{nN}$. Consider the overall states of a multi-agent system as a vector $\vec{x}=[\vec{x}_1^T,\vec{x}_2^T]^T$. The stacked vector $f(\vec{x})=[f^T_1(p_1,v_1),f^T_2(p_2,v_2), ...,f^T_N(p_N,v_N)]^T \in \mathbb{R}^{nN}$ consists of nonlinearities of all agents, the overall control input gain matrix is $g(\vec{x})=diag(g_1(p_1,v_1),g_2(p_2,v_2), ...,g_N(p_N,v_N)) \in \mathbb{R}^{nN\times nN}$, the stacked control input vector is $u=[u_1^T, u_2^T,..., u_N^T ]^T \in \mathbb{R}^{nN}$, and the stacked disturbance vector is $w=[w_1^T,w_2^T,..., w_N^T]^T \in \mathbb{R}^{nN}$.

The leader’s dynamics is given by
\begin{equation}\label{eq3_1}
\begin{split}
&\dot{p}_l=v_l,\\
&\dot{v}_l=f_l(p_l,v_l),
\end{split}
\end{equation}
where $p_{l}, v_{l}\in \mathbb{R}^n $ denote the position and velocity states of the leader. Here we use some standard assumptions \cite{ge2004adaptive,ge2003adaptive,fischer2013saturated,shin2018neural,patre2008asymptotic,wang2008adaptive,wen2017neural}
\begin{assumption}\label{ass1}
 The trajectory of the leader and its first three time-derivatives are bounded, i.e. $p_l,v_l,f_l,\dot{f}_l \in  \mathcal{L}_\infty $. The bounds of $p_l,v_l,f_l,\dot{f}_l$ are considered to be unknown.
\end{assumption}

\begin{assumption}\label{ass2}
The gain matrix $g_i(p_i,v_i)$ is a symmetric, positive definite matrix and its inverse satisfies the following inequality
\begin{equation}\label{eq_ass2}
\underline{g} ||x||^2 \le x^T g_i^{-1} x \le \bar{g} ||x||^2,
\end{equation}
where $x\in\mathbb{R}^n$ and constants $\underline{g}$ and  $\bar{g}$ are unknown positive constants.
\end{assumption}

\begin{remark}\label{rem1}
Assumption \ref{ass2} states a condition for the controllability of the system (\ref{eq1}). These bounds are considered for analytical purposes, and their exact values are considered to be unknown for control design. We removed an a priori knowledge about $\underline{g}$ in our control design. Its exact value is not required and this consideration distinguishes our work from \cite{IET,chendelay,ge2004adaptive,IEEE_TAES}. This assumption is not hard to satisfy as the large class of nonlinear systems, described by Euler–Lagrange formulation, meets this assumption \cite{Lewisbook,patre2008asymptotic}.

\end{remark}

\begin{assumption}\label{ass3}
The disturbance and the first two time-derivatives (i.e. $w_i,\dot{w}_w,\ddot{w}_i$) are considered to be bounded with unknown bounds.
\end{assumption}

Let us define the error vector for $i$-th agent as \cite{shi2019adaptive}
\begin{equation}\label{eq5_1}
\begin{split}
&e_{i}=\sum_{j \in N_i} a_{ij} (p_i-p_j-d_i+d_j)+ b_i(p_i-p_l-d_i), 
\end{split}
\end{equation}
where the constant vectors $d_{i} \in \mathbb{R}^n, i  \in \{1,..., N\} $ represents the desired relative position between agent $i$ and the leader. Note that we define $d_{ij} = d_i-d_j$, representing the desired relative position between agents $i$ and $j$.  For the agent $i$, define $\delta_i$ as $\delta_i \triangleq \dot{e}_i$
\begin{equation}\label{eq6_1}
\begin{split}
\delta_i =\sum_{j \in N_i} a_{ij} (v_i-v_j)+ b_i(v_i-v_l). 
\end{split}
\end{equation}

The control objective is to design a robust, distributed controller for each agent using local information that ensures all agents achieve the desired formation and maintain the leader's velocity: 
\begin{equation}\label{eq6}
\begin{split}
   & \text{lim}_{t \rightarrow \infty} ||p_i-p_l-d_i|| = 0, \\
&    \text{lim}_{t \rightarrow \infty} ||v_i-v_l|| = 0.
\end{split}
\end{equation}

  \begin{figure}[t]
    \centering
    \includegraphics[scale=.8]{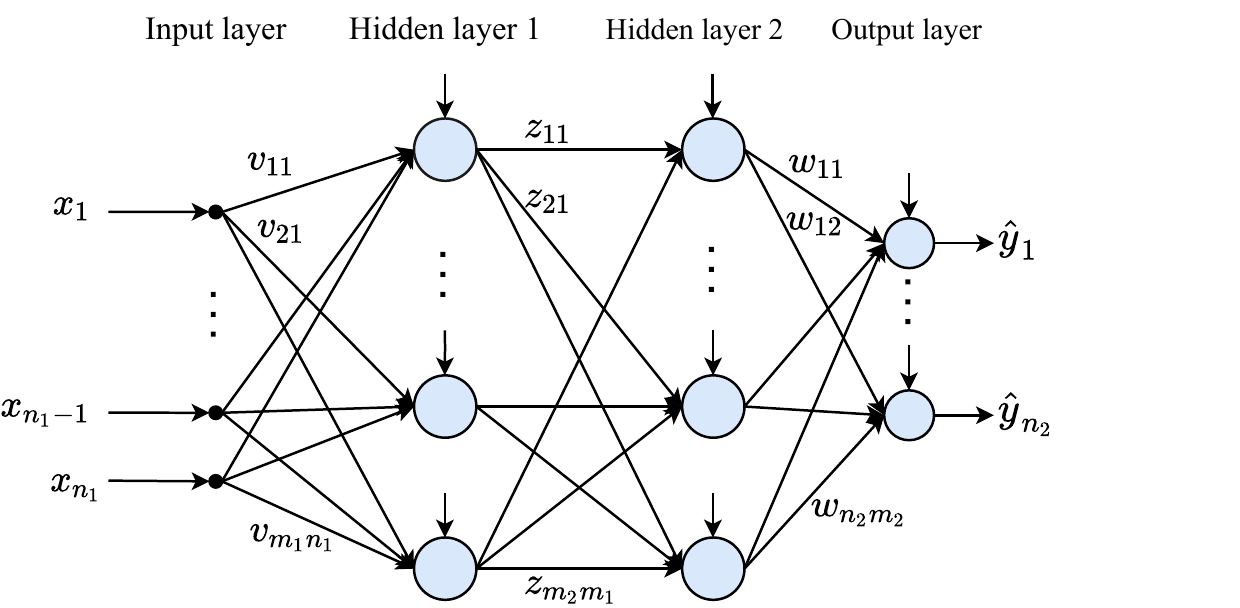} 
  \caption{Three-layer neural network architecture.}
   \label{fig1}
\end{figure}  

      \subsection{Neural Networks}\label{sub_d}

Based on the universal approximation theorem, for a smooth nonlinear function of $y(\mathcal{X})$ over a compact set of $\Omega_{\mathcal{X}}$ there exists a three-layer NN \cite{hornik1989multilayer}, such that 
 \begin{equation}\label{eq5}
     {y}( \mathcal{X}) = W^T \sigma_1 ( Z^T \sigma_2 (V^T \mathcal{X})) +\epsilon,
 \end{equation}
where $V$, $Z$, and $W$ are ideal NN weights matrices, $\sigma_1$ and $\sigma_2$ are the activation functions, and $\epsilon$ is the approximation error.
Here we establish some standard assumptions to estimate the nonlinearity \cite{lewis,Lewisbook,shin2018neural}.

\begin{assumption}\label{ass4}
 The unknown ideal neural network weights matrices $V$, $Z$ and $W$ are bounded with a fixed bound such that
 \begin{equation}
 \begin{split}
     ||V||_F & \le V_m, \\
     ||Z||_F & \le Z_m, \\
     ||W||_F & \le W_m.
 \end{split}
 \end{equation}
\end{assumption}

\begin{assumption}\label{ass5}
 For the vector $\mathcal{X} \in \mathcal{L}_{\infty}$, the approximation error $\epsilon$, and its first and second time-derivatives are bounded with a fixed bound.
\end{assumption}

In this paper, we use a three-layer NN, as shown in Fig. \ref{fig1}, to approximate an unknown nonlinear function over a compact set $\Omega_{\mathcal{X}}$.
 The first hidden layer has $m_1$ neurons and the second hidden layer has $m_2$ neurons. Considering $\hat{y} = [\hat{y}_{1}, ..., \hat{y}_{n_2}]^T \in \mathbb {R}^{n_2}$, the output of neural network (Fig. \ref{fig1}) can be written as 
 \begin{equation}\label{eq8}
     \hat{y} = \hat{W}^T \sigma_1 ( \hat{Z}^T \sigma_2 (\hat{V}^T \mathcal{X})),
 \end{equation}
where $\mathcal{X}= [1, \bar{\mathcal{X}}^T]^T  \in \mathbb{R}^{n_1+1}$ with $ \bar{\mathcal{X}}^T = [x_1,...,x_{n_1}]^T \in \mathbb{R}^{n_1}$, the neural network weights matrices $\hat{V} \in \mathbb{R}^{(n_1+1) \times m_1}$, $\hat{Z} \in \mathbb{R}^{(m_1+1) \times m_2}$, and $\hat{W} \in \mathbb{R}^{(m_2+1) \times n_2} $. 
We use $\sigma_1(.) \in \mathbb{R}^{m_2+1} $ and $\sigma_2(.) \in \mathbb{R}^{m_1+1}$ to denote the activation functions of hidden layers.
We consider the sigmoid function for both hidden layers, i.e., for a scalar $s_i$, one has $\sigma_k(s_i) = \frac{1}{1+e^{-s_i}}$ for $k\in\{1,2\}$. The activation function of the output layer is considered to be linear. The ideal neural network weights matrices $V$, $Z$, and $W$ are defined as follows:
\begin{equation}\label{eq13___1}
\small{
\begin{array}{lr}
V,Z,W= \arg \min_{\hat{V},\hat{Z},\hat{W} } \big{\{}\sup_{ \mathcal{X} \in \Omega_{\mathcal{X}}} ||y(\mathcal{X}) - \hat{y}(\mathcal{X}) ||\big{\}}.
\end{array}
}
\end{equation}

Let us define $\tilde{y}= y- \hat{y}$, $\tilde{\sigma}_k= \sigma_k- \hat{\sigma}_k$, $\tilde{V}= V- \hat{V}$, $\tilde{Z}= Z- \hat{Z}$, and $\tilde{W}= W- \hat{W}$. Let us define ${\sigma}_1 \triangleq \sigma_1({Z}^T  {\sigma}_2) $, with ${\sigma}_2 \triangleq {\sigma}_2({V}^T{\mathcal{X}})$ and $\hat{\sigma}_1 \triangleq \sigma_1(\hat{Z}^T  \hat{\sigma}_2) $  with $\hat{\sigma}_2 \triangleq {\sigma}_2(\hat{V}^T{\mathcal{X}}) $. Moreover, we express $\hat{\sigma}_1' \triangleq  \frac{d\sigma_1(s)}{ds}|_{s= \hat{Z}^T \hat{\sigma}_2 }$ and $\hat{\sigma}_2' \triangleq  \frac{d\sigma_2(s)}{ds}|_{s= \hat{V}^T \mathcal{X} }$.

\begin{lemma}\label{lem2}
The estimation error satisfies the following
\begin{equation}
\begin{split}
y-\hat{y} =& \tilde{W} ^T \big{(} \hat{\sigma}_1 - \hat{\sigma}_1' \hat{Z}^T {\hat{\sigma}_2}- \hat{\sigma}_1' \hat{Z}^T  \hat{\sigma}_2' \hat{V}^T{\mathcal{X}} \big{)}  \\  
& + \hat{W} ^T  \hat{\sigma}_1' \tilde{Z} ^T \big{(} \hat{\sigma}_2 -  \hat{\sigma}_2' \hat{V}^T \mathcal{X} \big{)}  \\ 
 & +\hat{W} ^T \big{(}  \hat{\sigma}_1' \hat{Z} ^T \hat{\sigma}_2' \tilde{V}^T \mathcal{X} \big{)}  + \bar{\epsilon},
\end{split}
\end{equation}
with  
\begin{equation}\label{eq12_1}
\small{
\begin{aligned}
\bar{\epsilon}=    &\hat{W}^T \hat{\sigma}_1'  Z^T  \textit{O}_2(.)+\hat{W}^T \hat{\sigma}_1'  \tilde{Z}^T \hat{\sigma}_2' V^T \mathcal{X} \\ 
& + [\tilde{W}^T (\hat{\sigma}_1' {Z}^T\sigma_2 ) + \tilde{W} \hat{\sigma}_i'\hat{Z}^T \hat{\sigma}_2'\hat{V}^T \mathcal{X} + W^T\textit{O}_1(.)  +\epsilon ],
\end{aligned}}
\end{equation}
where $\textit{O}_1(.)$ and $\textit{O}_2(.)$  are higher-order terms of Taylor series expansion of $ {\sigma}_1(Z^T \sigma_2 (V^T{\mathcal{X}})) $ at $\hat{Z}^T \sigma_2(\hat{V}^T{\mathcal{X}})$ and $ {\sigma}_2({V}^T{\mathcal{X}}) $ at $\hat{V}^T{\mathcal{X}}$, respectively.
\end{lemma}

\begin{proof}
Following the procedure in \cite{Lewisbook,park2005direct}, from (\ref{eq5}) and (\ref{eq8}), one can write 
\begin{equation}\label{eq11}
    \begin{split}
    y-\hat{y} &= {W}^T \sigma_1 ( {Z}^T \sigma_2 ({V}^T \mathcal{X})) - \hat{W}^T \sigma_1 ( \hat{Z}^T \sigma_2 (\hat{V}^T \mathcal{X}))  +\epsilon. 
\end{split} 
\end{equation}
By adding and subtracting $W^T \hat{\sigma}_1 + \hat{W}^T\tilde{\sigma}_1 $, and rearranging (\ref{eq11}), one has
\begin{equation}\label{eq12}
    \begin{split}
    y-\hat{y} &= \tilde{W}^T \hat{\sigma}_1+\tilde{W}^T \tilde{\sigma}_1 +\hat{W}^T \tilde{\sigma}_1  +\epsilon. 
\end{split} 
\end{equation}
Taylor series expansions of ${\sigma}_1(Z^T \sigma_2 (V^T{\mathcal{X}} )) $ at $\hat{Z}^T \sigma_2(\hat{V}^T{\mathcal{X}})$ and $ {\sigma}_2({V}^T{\mathcal{X}}) $ at $\hat{V}^T{\mathcal{X}}$ is 
\begin{equation}\label{eq13}
{\small \begin{aligned}
    {\sigma}_1(Z^T \sigma_2 (V^T{\mathcal{X}})) = & {\sigma}_1(\hat{Z}^T \sigma_2 (\hat{V}^T{\mathcal{X}})) \\  &+ \hat{\sigma}_1'[Z^T\sigma_2(V^T{\mathcal{X}})  - \hat{Z}^T\sigma_2(\hat{V}^T{\mathcal{X}})]  +  \textit{O}_1(.) \\ =&
    \hat{\sigma}_1 + \hat{\sigma}_1'[Z^T\sigma_2  - \hat{Z}^T\hat{\sigma}_2]  +  \textit{O}_1(.),
\end{aligned}}
\end{equation}
\begin{equation}\label{eq14}
\begin{split}
\sigma_2 (V^T{\mathcal{X}}) = \hat{\sigma}_2 +\hat{\sigma}_2' \tilde{V}^T{\mathcal{X}} + \textit{O}_2(.).
\end{split}
\end{equation}
Note that $\textit{O}_1(.)$ and $\textit{O}_2(.)$ are shortened notations for  $\textit{O}_1(Z^T\sigma_2(V^T{\mathcal{X}}) - \hat{Z}^T\sigma_2(\hat{V}^T{\mathcal{X}}))$ and $\textit{O}_2(\tilde{V}^T{\mathcal{X}})$, respectively.   From (\ref{eq13}) and  (\ref{eq14}), it is evident that 
\begin{equation}\label{eq15}
\begin{split}
    \tilde{\sigma}_1 &= \hat{\sigma}_1'[Z^T\sigma_2  - \hat{Z}^T\hat{\sigma}_2]  +  \textit{O}_1(.),\\ 
    \tilde{\sigma}_2 & = \hat{\sigma}_2' \tilde{V}^T{\mathcal{X}} + \textit{O}_2(.).
\end{split}
\end{equation}
Substituting (\ref{eq15}) in  (\ref{eq12}), one has
\begin{equation}\label{eq16}
\begin{split}
  y-\hat{y} &= \tilde{W}^T \big{(} \hat{\sigma}_1+ \hat{\sigma}_1'[Z^T\sigma_2  - \hat{Z}^T\hat{\sigma}_2]  +  \textit{O}_1(.)   \big{)} + \hat{W}^T \tilde{\sigma}_1  +\epsilon. 
\end{split}
\end{equation}
Rearranging (\ref{eq16}), yields
\begin{equation}\label{eq17}
\begin{split}
  y-\hat{y} =& \tilde{W}^T \big{(} \hat{\sigma}_1- \hat{\sigma}_1' \hat{Z}^T\hat{\sigma}_2 +  \hat{\sigma}_1' \hat{Z}^T\tilde{\sigma}_2  \big{)}   + \hat{W}^T \tilde{\sigma}_1   \\& 
  +[\tilde{W}^T (\hat{\sigma}_1' [\tilde{Z}^T\sigma_2 +\hat{Z}^T\hat{\sigma}_2]+ \textit{O}_1(.) ) +\epsilon ]. 
\end{split}
\end{equation}
From (\ref{eq15}), (\ref{eq17}) can be written as
\begin{equation}\label{eq18}
{\small
\begin{aligned}
 y-\hat{y} = & \tilde{W}^T \big{(} \hat{\sigma}_1- \hat{\sigma}_1' \hat{Z}^T\hat{\sigma}_2 - \hat{\sigma}_1' \hat{Z}^T\hat{\sigma}_2' \hat{V}^T \mathcal{X}     \big{)}   + \hat{W}^T \tilde{\sigma}_1    \\ 
 & +[\tilde{W}^T (\hat{\sigma}_1' [\tilde{Z}^T\sigma_2 +\hat{Z}^T\hat{\sigma}_2+\hat{Z}^T\hat{\sigma}_2' V^T{\mathcal{X}}+\hat{Z}^T \textit{O}_2(.)] 
 \\ & +\textit{O}_1(.) ) +\epsilon ]. 
\end{aligned}}
\end{equation}
Following the same procedure for $\hat{W}^T \tilde{\sigma}_1$ and from (\ref{eq15}), one has
\begin{equation}\label{eq19}
\small{\begin{aligned}
  y-\hat{y} =&  \tilde{W}^T \big{(} \hat{\sigma}_1- \hat{\sigma}_1' \hat{Z}^T\hat{\sigma}_2 - \hat{\sigma}_1' \hat{Z}^T\hat{\sigma}_2' \hat{V}^T \mathcal{X}     \big{)}  \\
 &  + \hat{W}^T \hat{\sigma}_1' [ \tilde{Z}^T \tilde{\sigma}_2+\tilde{Z}^T \hat{\sigma}_2] +  \hat{W}^T \hat{\sigma}_1' [ \hat{Z}^T \tilde{\sigma}_2] + \hat{W}^T \textit{O}_1(.) \\ 
&  +[\tilde{W}^T (\hat{\sigma}_1' [\tilde{Z}^T\sigma_2 +\hat{Z}^T\hat{\sigma}_2+\hat{Z}^T\hat{\sigma}_2' V^T{\mathcal{X}}+\hat{Z}^T \textit{O}_2(.)] \\& + \textit{O}_1(.) ) +\epsilon ]. 
\end{aligned}}
\end{equation}
Substituting (\ref{eq15}) in (\ref{eq19}) leads to
\begin{equation}\label{eq20}
{\small{
 \begin{aligned}
 y-\hat{y} =  & \tilde{W}^T \big{(} \hat{\sigma}_1- \hat{\sigma}_1' \hat{Z}^T\hat{\sigma}_2 - \hat{\sigma}_1' \hat{Z}^T\hat{\sigma}_2' \hat{V}^T \mathcal{X}     \big{)} \\
 & +\hat{W}^T \hat{\sigma}_1'  \tilde{Z}^T ( \hat{\sigma}_2- \hat{\sigma}_2'\hat{V}^T \mathcal{X})   +\hat{W}^T \hat{\sigma}_1'  \tilde{Z}^T  \textit{O}_2(.) 
 \\ & +\hat{W}^T \hat{\sigma}_1'  \tilde{Z}^T \hat{\sigma}_2' V^T \mathcal{X} + 
    \hat{W}^T \hat{\sigma}_1'  \hat{Z}^T \hat{\sigma}_2' \tilde{V}^T \mathcal{X}  
    \\ & + 
    \hat{W}^T \hat{\sigma}_1'  \hat{Z}^T \textit{O}_2(.)+  W^T\textit{O}_1(.) + 
    \tilde{W}^T (\hat{\sigma}_1' [\tilde{Z}^T\sigma_2 
    \\ & +\hat{Z}^T\hat{\sigma}_2+\hat{Z}^T\hat{\sigma}_2' V^T{\mathcal{X}}+\hat{Z}^T \textit{O}_2(.)] )  +\epsilon . 
 \end{aligned}}}
\end{equation}
From (\ref{eq14}) and (\ref{eq20}), one can derive $\bar{\epsilon}$ as (\ref{eq12_1}).
 \hfill  $\blacksquare$
\end{proof}

The following lemma provides a bound for $\bar{\epsilon}$.

\begin{lemma}\label{lemma3}
The term $\bar{\epsilon}$ (\ref{eq12_1}) satisfies the following inequality
\begin{equation}\label{eq23}
\begin{split}
    ||\bar{\epsilon}||& \le \Gamma \mu, \\
\end{split}
\end{equation}
where $\Gamma$ is an unknown positive constant, and $\mu $ is defined as follows:
\begin{equation}\label{eq25_1}
\begin{split}
   \mu = &  ||\hat{W}||_F  +||\hat{W}||_F ||\mathcal{X}||+||\hat{W}||_F ||\hat{V}||_F||\mathcal{X}||    \\ 
& +||\hat{W}||_F ||\hat{Z}||_F  ||\mathcal{X}|| +   ||\hat{Z}||_F ||\hat{V}||_F ||\mathcal{X}|| \\  
&  + ||\hat{W}||_F ||\hat{Z}||_F ||\hat{V}||_F ||\mathcal{X}|| +||\hat{Z}||_F+1.
\end{split} 
\end{equation}
\end{lemma} 

\begin{proof}
Following the procedure in \cite[Lemma 4.3.1]{Lewisbook}, \cite{park2005direct}, as well as using the fact that $\sigma_k(.)$ and its derivative are bounded, and from (\ref{eq15}), one can show that 
\begin{equation}\label{eq24}
\begin{split}
    ||\textit{O}_2(.)||& \le c_0 +c_1||\mathcal{X}||+c_2 ||\hat{V}||_F||\mathcal{X}||, \\
    ||\textit{O}_1(.)||& \le c_3 +c_4||\hat{Z}||_F,\\
\end{split}
\end{equation}
where $c_i$, $i \in \{0,..,4\}$ are positive constants. Using Assumption \ref{ass4}, and by substituting  (\ref{eq24}) in (\ref{eq12_1}), one has
\begin{equation}\label{eq25}\small{
\begin{aligned}
     ||\bar{\epsilon}||\le &||\hat{W}^T \hat{\sigma}_1'  Z^T  || \ ||\textit{O}_2(.)||+ ||\hat{W}^T \hat{\sigma}_1'  \tilde{Z}^T \hat{\sigma}_2' V^T \mathcal{X}|| \\ 
& + ||\tilde{W}^T \hat{\sigma}_1' {Z}^T\sigma_2 || + ||\tilde{W} \hat{\sigma}_i'\hat{Z}^T \hat{\sigma}_2'\hat{V}^T \mathcal{X} || \\  
& +|| W^T\textit{O}_1(.)  +\epsilon||  \\ 
\le & d_0 ||\hat{W}||_F ( c_0 +c_1||\mathcal{X}||+c_2 ||\hat{V}||_F||\mathcal{X}||) \\ 
& + d_1||\hat{W}||_F \ ||\mathcal{X}||   +d_2||\hat{W}||_F ||\hat{Z}||_F  ||\mathcal{X}|| + d_3 ||\hat{W}||_F \\ 
& + d_4  ||\hat{Z}||_F ||\hat{V}||_F ||\mathcal{X}|| 
 + d_5 ||\hat{W}||_F ||\hat{Z}||_F ||\hat{V}||_F ||\mathcal{X}|| \\ & +d_6 (c_3 +c_4||\hat{Z}||_F) +d_7, 
\end{aligned}}
\end{equation}
where $d_i$, $i \in \{1,..,7\}$ are positive constants. Rearranging  (\ref{eq25}) yields
\begin{equation}\label{eq26}\small{
\begin{aligned}
     ||\bar{\epsilon}|| \le &  (d_0 c_0+d_3) ||\hat{W}||_F  +(c_1+d_1)||\hat{W}||_F ||\mathcal{X}|| \\ 
& +c_2 ||\hat{W}||_F ||\hat{V}||_F||\mathcal{X}||  +d_2||\hat{W}||_F ||\hat{Z}||_F  ||\mathcal{X}|| \\ 
&  + d_4  ||\hat{Z}||_F ||\hat{V}||_F ||\mathcal{X}|| + d_5 ||\hat{W}||_F ||\hat{Z}||_F ||\hat{V}||_F ||\mathcal{X}|| \\  
& +d_6 c_3+d_7 +c_4||\hat{Z}||_F). 
\end{aligned}}
\end{equation}
From (\ref{eq26}), one has
\begin{equation}\label{eq28}\small{
\begin{aligned}
     ||\bar{\epsilon}|| \le & \Gamma( ||\hat{W}||_F  +||\hat{W}||_F ||\mathcal{X}||+||\hat{W}||_F ||\hat{V}||_F||\mathcal{X}||    \\ 
& +||\hat{W}||_F ||\hat{Z}||_F  ||\mathcal{X}|| +   ||\hat{Z}||_F ||\hat{V}||_F ||\mathcal{X}|| \\  
&  +||\hat{W}||_F ||\hat{Z}||_F ||\hat{V}||_F ||\mathcal{X}|| +||\hat{Z}||_F+1), 
\end{aligned}}
\end{equation}
with 
\begin{equation}\label{eq30}\small{
\begin{aligned}
    \Gamma \triangleq max\{d_0 c_0+d_3,c_1+d_1, c_2,d_2, d_4, d_5, d_6 c_3+d_7 +c_4 \}.
\end{aligned}}
\end{equation}
From (\ref{eq30}), it is straightforward to verify that selecting $\mu$ as (\ref{eq25_1}) implies (\ref{eq23}). 
\hfill $\blacksquare$
\end{proof}

\begin{remark}\label{rem2}
Lemma \ref{lemma3} provides an upper bound for $\bar{\epsilon}$ using a three-layer neural network. Note that this error is bounded as $||\bar{\epsilon}|| \le \Gamma \mu$, where $\Gamma$ is considered to be an unknown constant. Therefore, one can write it as $||\bar{\epsilon}|| \le \bar{\epsilon}_M$, where $\bar{\epsilon}_M$ considered to be unknown. Note that we consider this bound for analytical purposes, and its exact value is not required in our control design.
\end{remark}

\begin{remark}\label{rem3}
For any sufficiently smooth function over a compact set, the result of the Stone-Weierstrass Theorem stipulates that the function can be approximated using enough number of neurons, and the universal approximation property of NN guarantees the boundedness of the approximation error \cite{stone1948generalized,patre2008asymptotic,lewis,hornik1989multilayer}.
\end{remark}

\section{Controller Design}
Let us define $\vec{\delta} = [\delta^T_1, ...,\delta_N^T] ^T \in \mathbb{R}^{nN}$ and $\vec{e} = [e^T_1, ...,e_N^T] ^T \in \mathbb{R}^{nN}$. Then, one can write (\ref{eq5_1}) and (\ref{eq6_1}) as
\begin{equation}\label{eq31}
    \begin{split}
        \vec{e} &= [(L+B) \otimes I_n ] (\vec{x}_1 -\vec{p}_l - d), \\
         \vec{\delta} &= [(L+B) \otimes I_n ] (\vec{x}_2 -\vec{v}_l),
    \end{split}
\end{equation}
where $\vec{p}_l = \mathbf{1}_N \otimes p_l$, $d=[d_1^T,...,d_N^T]^T$ and $\vec{v}_l = \mathbf{1}_N \otimes v_l$. 
Let us define an auxiliary variable of $\zeta_i = \delta_i+k_1 e_i$ as in \cite{lewis}, where $k_1$ is a positive constant gain. Then the stacked vector $\zeta = [\zeta^T_1, ...,\zeta_N^T] ^T \in \mathbb{R}^{nN}$ is defined as 
\begin{equation}\label{eq33}
    \begin{split}
       \zeta =  \vec{\delta} + k_1 \vec{e}.
    \end{split}
\end{equation}
One can derive the time-derivative of (\ref{eq33}) as
\begin{equation}\label{eq34}
\dot{\vec{\zeta}} = [(L+B)\otimes I_n] \big{(}f(\vec{x}) + g(\vec{x}) u + w- \vec{f}_l \big{)}  + k_1 \delta , 
\end{equation}
where $\vec{f}_l = \mathbf{1}_N \otimes f_l$.

\subsection{RISE Feedback}

Let us define the filtered error as 
\begin{equation}\label{eq40}
    r = \dot{\zeta} + k_2\zeta,
\end{equation}
where $k_2$ is a positive constant. The time-derivative of (\ref{eq40}) is given by
\begin{equation}\label{eq41}
    \dot{r} = \ddot{\zeta} + k_2\dot{\zeta}.
\end{equation}
From Assumption \ref{ass2}, Remark \ref{rem1}, Lemma \ref{lem_l+b}, and (\ref{eq41}), one can express the following equation similar to \cite{yang2015robust} as
\begin{equation}\label{eq44}
   \mathcal{G}(\vec{x}) Q \dot{r} = \mathcal{G}(\vec{x}) Q (\ddot{\zeta} +  k_2\dot{\zeta}),
\end{equation}
where $\mathcal{G}(\vec{x}) \triangleq g^{-1}(\vec{x})$ and $Q \triangleq [(L+B)\otimes I_n]^{-1}$. By adding and subtracting $\frac{1}{2} \dot{\mathcal{G}}(\vec{x}) Qr + k_2\zeta$, and from (\ref{eq31})-(\ref{eq34}),  (\ref{eq44}) can be rewritten as 
\begin{equation}\label{eq45}
\begin{split}
   \mathcal{G}(\vec{x}) Q \dot{r} =& \mathcal{G}(\vec{x}) \Big{(}  (k_2+k_1) f(\vec{x}) + (k_2+k_1) g(         \vec{x}) u \\ 
   &+   (k_2+k_1) w - (k_2+k_1)\vec{f}_l + k_1 k_2 (\vec{x}_2 -\vec{v_l}) \\  
   &  + \dot{f}(\vec{x}) +  \dot{g}(\vec{x}) u +{g}(\vec{x}) \dot{u}  + \dot{w}  - \dot{\vec{f}}_l \Big{)} \\ 
   &   + \frac{1}{2}\dot{\mathcal{G}}(\vec{x}) Q r + k_2\zeta - \frac{1}{2}\dot{\mathcal{G}}(\vec{x}) Q r - k_2\zeta \\ 
    = & -\frac{1}{2}\dot{\mathcal{G}}(\vec{x}) Qr - k_2\zeta + \dot{u} +N_1 + N_2,
\end{split} 
\end{equation}
where 
\begin{equation}\label{eq46}
\begin{split}
    \dot{\mathcal{G}}(\vec{x}) &=  \frac{\partial \mathcal{G} }{\partial{\vec{x}_1}} \vec{x}_2  + \frac{\partial \mathcal{G} }{\partial{\vec{x}_2}} \big{(}f(\vec{x}) + g(\vec{x})u + w\big{)},  \\ 
     \dot{{f}}(\vec{x}) &=  \frac{\partial f }{\partial{\vec{x}_1}} \vec{x}_2  + \frac{\partial f }{\partial{\vec{x}_2}} \big{(}f(\vec{x}) + g(\vec{x})u + w\big{)},  \\
        \dot{\vec{f}}_l & =  \mathbf{1}_N \otimes \big{(}\frac{\partial f_l }{\partial{{p}_l}} {v}_l  + \frac{\partial f_l }{\partial{{v}_l}} f_l  \big{)}.
     \end{split}
\end{equation}
Using the fact that $\mathcal{G}(\vec{x}) g(\vec{x}) = I_{nN}$, and from Assumption \ref{ass2}, $N_1$ and $N_2$ are given as follows:
\begin{equation}\label{eq47}
\begin{split}
  N_1=& (k_1+ k_2) (\mathcal{G}(\vec{x})  f(\vec{x})+  u)  + \mathcal{G}(\vec{x})\frac{\partial f }{\partial{\vec{x}_1}} \vec{x}_2   \\  
& + \mathcal{G}(\vec{x})\frac{\partial f }{\partial{\vec{x}_2}} \big{(}f(\vec{x}) 
    + g(\vec{x})u \big{)}  + k_2 \vec{x}_2 + k_2k_1 (\vec{x}_1-d) 
\\  
 &   + \frac{1}{2} \Big{(} \frac{\partial \mathcal{G} }{\partial{\vec{x}_1}} \vec{x}_2   + \frac{\partial \mathcal{G} }{\partial{\vec{x}_2}} \big{(}f(\vec{x})  + g(\vec{x})u \big{)} \Big{)} \big{(} f(\vec{x}) - g(\vec{x})u \\  
  &   + (k_1 + k_2) \vec{x}_2 + k_1 k_2 (\vec{x}_1 -d)  \big{)} + k_1k_2 \mathcal{G}(\vec{x}) \vec{x}_2,
\end{split}
\end{equation}
\begin{equation}\label{eq48}
\begin{split}
  N_2=& (k_1 + k_2) (\mathcal{G}(\vec{x})  w -\mathcal{G}(\vec{x})  \vec{f}_l) + \mathcal{G}(\vec{x}) ( \dot{w}- \dot{\vec{f}}_l) -k_1k_2 \vec{p}_l
   \\& +  \mathcal{G}(\vec{x}) \frac{\partial f }{\partial{\vec{x}_2}} w - k_1k_2  \mathcal{G}(\vec{x})  {\vec{v}}_l    -   \frac{\partial \mathcal{G} }{\partial{\vec{x}_2}} w g(\vec{x}) u -k_2 \vec{v}_l\\ 
  & +\frac{1}{2} \Big{(} \frac{\partial \mathcal{G} }{\partial{\vec{x}_1}} \vec{x}_2  + \frac{\partial \mathcal{G} }{\partial{\vec{x}_2}} \big{(}f(\vec{x})   + g(\vec{x})u \big{)} \Big{)}  (w-\dot{\vec{v}}_l). 
\end{split} 
\end{equation}
Note that $N_1=[N^T_{1_1},...,N^T_{1_N}]^T\in \mathbb{R}^{nN}$ and $N_2=[N^T_{2_1},...,N^T_{2_N}]^T\in \mathbb{R}^{nN}$ consist of unknown terms. The difference between $N_{1_i}$ and $N_{2_i}$ is that, for each agent, we use NN to approximate the variable $N_{1_i}$ while we use a time-varying robustifying term to compensate for $N_{2_i}$. The variable $N_{1_i}$ is a function of available filtered error $\zeta_i$, states, and agent control input. Note that feeding the agent control input into the input layer of the NN causes the circular design problem. To avoid this problem, similar to \cite{park2005direct,shin2018neural}, we construct the NN from the NN weights matrices, states of the system dynamics, the robustifying term, and $\zeta_i$. This consideration distinguishes our results from \cite{yang2015robust}. 

Considering (\ref{eq45}), the NN-based control law is derived as
\begin{equation}\label{eq42_1}
\begin{split}
\dot{u} = -(k_3+k_4) \dot{\zeta} - k_2(k_3+k_4) {\zeta}  -  \hat{N}_1 - \kappa(t)\text{sgn}(\zeta(t)),
\end{split} 
\end{equation}
where $k_3$, $k_4$ are positive constants, $\hat{N}_1$ is the output of NN, and $\kappa(t)=[diag(\kappa_1,...,\kappa_N)] \otimes I_{n} \in \mathbb{R}^{nN \times nN}$ is the time-varying gain. The last term in (\ref{eq42_1}) is responsible for compensating the unknown term of $N_2$ in (\ref{eq45}) and the NN approximation error. For the agent $i$, let us define $\hat{N}_{2_i}  \triangleq \kappa_i \text{sgn}(\zeta_i)$. The block diagram of proposed control law is illustrated in Fig. \ref{fig2}.

We choose the following NN input vector for each agent:  
\begin{equation}\label{eq44_1}
\begin{split}
\bar{\mathcal{X}}_i =& [ p^T_i, v^T_i, \zeta^T_i, ||\hat{V}_i||_F, ||\hat{Z}_i||_F, ||\hat{W}_i||_F, \kappa_i(t)]^T \in \mathbb{R}^{3n+4},
\end{split} 
\end{equation}
and $\mathcal{X}_i=[1,\bar{\mathcal{X}}_i ]^T \in \mathbb{R}^{3n+5}$. The NN for $i$-th agent is described as (\ref{eq8}), i.e., $y({\mathcal{X}_i}) \triangleq N_{1_i}$ and $\hat{y}({\mathcal{X}_i}) \triangleq \hat{N}_{1_i}$. Let $\tilde{V}=diag(\tilde{V}_i) \in \mathbb{R}^{(3n+5)N \times m_1 N}$,
$\tilde{Z}=diag(\tilde{Z}_i) \in \mathbb{R}^{(m_1+1)N \times m_2N}$, $\tilde{W}=diag(\tilde{W}_i) \in \mathbb{R}^{(m_2+1)N \times n N}$,
$\hat{V}=diag(\hat{V}_i)$, $\hat{Z}=diag(\hat{Z}_i)$, $\hat{W}=diag(\hat{W}_i)$,  $\hat{\sigma}_s={[{\hat{\sigma}}^T_1,...,{\hat{\sigma}}^T_N]} ^T$, and ${\hat{\sigma}'}_s=[{\hat{\sigma}'^T}_1,...,{\hat{\sigma}'^T}_N]^T,  s \in \{1,2\}$. Let us define $\mathcal{X}=[\mathcal{X}^T_1, ..., \mathcal{X}^T_N]^T$. 

Define $\hat{N}_1 =\hat{W}^T \sigma_1 ( \hat{Z}^T \sigma_2 (\hat{V}^T \mathcal{X}))$ and $\tilde{N}_1 =[\tilde{N}^T_{1_1}, ..., \tilde{N}^T_{1_N}]^T,$ with $\tilde{N}_{1_i}= {N}_{1_i}-\hat{N}_{1_i}$ as
\begin{equation}\label{eq_46_N}
    \begin{split}
       \tilde{N}_{1_i} & = \tilde{W}_i ^T \big{(} \hat{\sigma}_1 - \hat{\sigma}_1' \hat{Z}_i^T {\hat{\sigma}_2}- \hat{\sigma}_1' \hat{Z}_i^T  \hat{\sigma}_2' \hat{V}_i^T{\mathcal{X}_i} \big{)}  \\  
& + \hat{W}_i ^T  \hat{\sigma}_1' \tilde{Z}_i ^T \big{(} \hat{\sigma}_2 -  \hat{\sigma}_2' \hat{V}_i^T \mathcal{X}_i \big{)}  \\ 
 & +\hat{W}_i ^T \big{(}  \hat{\sigma}_1' \hat{Z}_i ^T \hat{\sigma}_2' \tilde{V}_i^T \mathcal{X}_i \big{)}.
    \end{split}
\end{equation}
For the sake of the notation simplicity, the index of $i$ and the arguments of $\sigma_1$ and $\sigma_2$ are dropped. From (\ref{eq40}), (\ref{eq45}), (\ref{eq42_1}), Lemma \ref{lem2}, and Remark \ref{rem2}, one has
\begin{equation}\label{eq43_1}
\begin{split}
\mathcal{G}(\vec{x}) Q \dot{r} =&  -\frac{1}{2}\dot{\mathcal{G}}(\vec{x}) Qr - k_2\zeta  -(k_3+k_4) {r}  - \kappa(t)\text{sgn}(\zeta) \\& +\tilde{N}_1+N_2 + \epsilon,
\end{split} 
\end{equation}
where $\epsilon =[\epsilon_1^T, ..., \epsilon^T_N]^T \in \mathbb{R}^{nN}$.

\begin{remark}\label{rem4}
Let $\vec{x} \in \Omega_{\vec{x}}$, where set $\Omega_{\vec{x}}$ is a compact set. Following the same procedure as in \cite{yang2015robust} and considering the Assumptions \ref{ass2} and \ref{ass3}, if the control input signal is bounded, i.e., $u\in L_\infty$, the variable $\dot{\vec{x}}$ is bounded. As the $\vec{p}_l$, $\vec{v}_l$ and $\vec{f}_l$ are bounded (Assumption \ref{ass1}), signals $\vec{e}, \delta$, and $\zeta$ remain bounded, as well as $\dot{\zeta}$, i.e., $\dot{\zeta}\in\Omega_{\dot{\zeta}}$ (see (\ref{eq34})). Moreover, considering $\kappa \in \Omega_\kappa$, $\zeta \in \Omega_\zeta$, $Y \in \Omega_Y$ with $Y=[||\hat{W}||_F, ||\hat{Z}||_F, ||\hat{V}||_F]^T$, and with the help of (\ref{eq42_1}), variable $\dot{u}$ is bounded, i.e., $\dot{u} \in L_\infty$.  Then, take the time-derivative of (\ref{eq44_1}), consider Assumptions \ref{ass1}, \ref{ass3}, (\ref{eq34}) and (\ref{eq40}). Then it follows that  $\dot{\mathcal{X}}$ remains bounded in the compact set of $\Omega_\mathcal{X}$ with $\Omega_\mathcal{X}  \triangleq \Omega_\vec{x} \times \Omega_\vec{u} \times \Omega_\kappa \times \Omega_Y$, where $\Omega_\vec{u}$ is the compact set in which $u$ is bounded.
\end{remark}

Let us define an agent's NN weights matrices tuning laws as 
{\small{\begin{equation}\label{eq_tune}
\begin{split}
    {\hat{W}}_i(t)=& \alpha_i \Big{(}k_2  \int_0^t \big{(} \hat{\sigma}_1 - \hat{\sigma}_1' \hat{Z}_i^T(s) {\hat{\sigma}_2} \\& 
    -\hat{\sigma}_1' \hat{Z}_i^T(s)  \hat{\sigma}_2' \hat{V}_i^T(s){\mathcal{X}_i(s)} \big{)} \zeta_i^T(s) ds, - \|\zeta_i\|_1 {\hat{W}}_i  \Big{)}  \\  
{\hat{Z}}_i(t)  =& \beta_i \Big{(} k_2  \int_0^t  \big{(} \hat{\sigma}_2 -  \hat{\sigma}_2' \hat{V}_i^T(s) \mathcal{X}_i \big{)} \zeta_i^T(s) \hat{W}^T_i(s)  \hat{\sigma}_1' ds  \\ & - \|\zeta_i\|_1 {\hat{Z}}_i  \Big{)} , \\ 
{\hat{V}}_i(t) =&  \gamma_i \big{(} k_2 \int_0^t \mathcal{X}_i(s)  \zeta^T_i(s) \hat{W}^T_i(s)   \hat{\sigma}_1' \hat{Z}^T_i(s) \hat{\sigma}_2' ds \\& - \|\zeta_i\|_1 {\hat{V}}_i  \big{)}, 
\end{split}\end{equation}}}\noindent
with $\alpha_i,\beta_i,\gamma_i$ being positive constants. Let us define $\alpha =diag(\mathbf{1}_{m_2+1}^T \alpha_1,...,\mathbf{1}_{m_2+1}^T \alpha_N) \in \mathbb{R}^{(m_2+1)N\times (m_2+1)N}$, $\beta =diag(\mathbf{1}_{m_1+1}^T \beta_1,...,\mathbf{1}_{m_1+1}^T \beta_N) \in \mathbb{R}^{(m_1+1)N\times (m_1+1)N} $ and $\gamma=diag(\mathbf{1}_{3n+5}^T \gamma_1,...,\mathbf{1}_{3n+5}^T \gamma_N) \in \mathbb{R}^{(3n+5)N\times (3n+5)N}$. Let the time-varying gain $\kappa_i(t)$ be defined as \cite{yang2015robust,fan2017asymptotic}
\begin{equation}\label{eq_kappa}
\begin{split}
    \dot{\kappa}_i(t) = r_i^T \text{sgn}(\zeta_i).
\end{split}\end{equation}
From (\ref{eq40}) and by taking the integral from (\ref{eq_kappa}), one has
\begin{equation}\label{eq_kappa_N}
\begin{split}
    {\kappa}_i(t) &=\int_0^t \big{(}\dot{\zeta}_i^T(s) \text{sgn}(\zeta_i(s)) + k_2{\zeta}_i^T(s) \text{sgn}(\zeta_i(s))\big{)}ds \\ & =  \sum _{j=1}^{n} \int_0^t \frac{d \zeta_{i_j}(s)}{ds}  \text{sgn}(\zeta_{i_j}(s))  ds +   \int_0^t k_2\| {\zeta}_i(s)\|_1 ds 
    \\ & =  \| {\zeta}_i(t)\|_1 - \| {\zeta}_i(0)\|_1+  k_2 \int_0^t \| {\zeta}_i(s)\|_1 ds .
\end{split}
\end{equation}

\begin{figure*}[t]
    \centering
    \includegraphics[scale=.77]{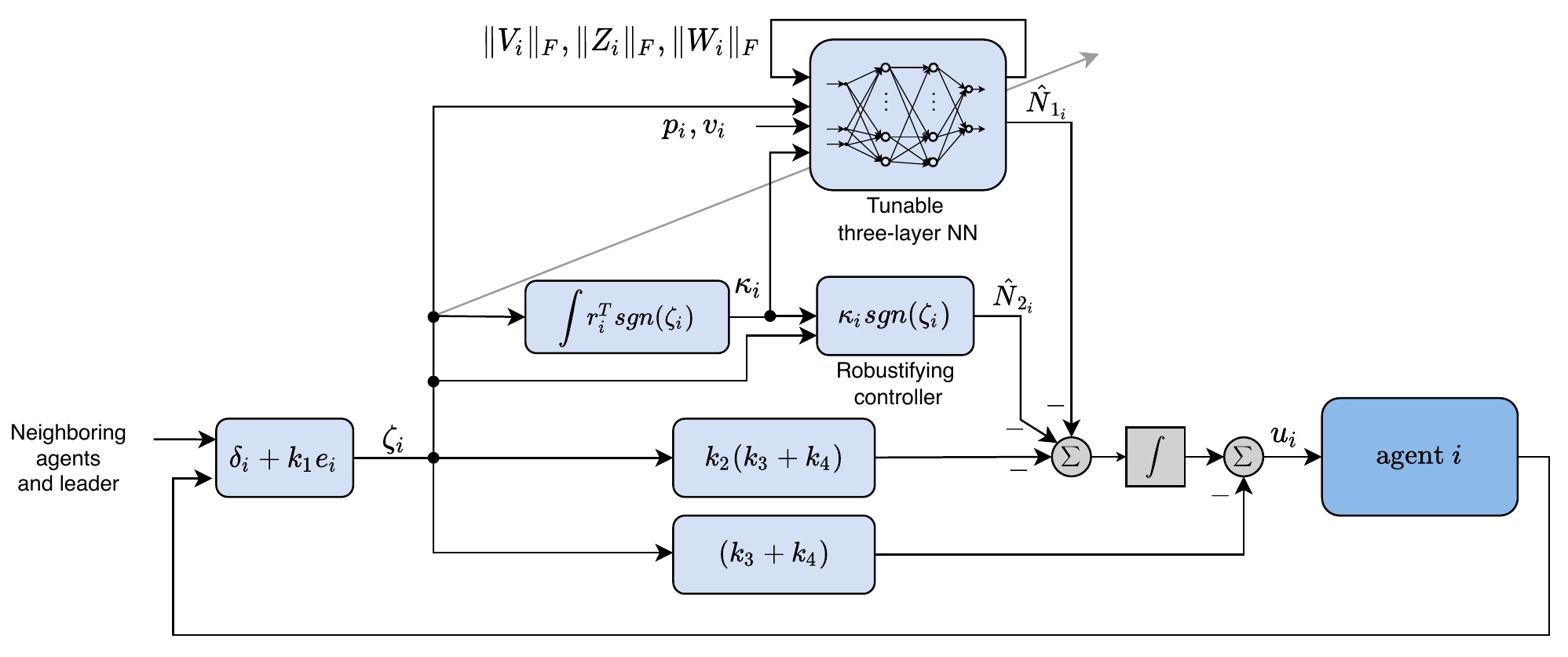}    
  \caption{The block diagram of proposed control law (\ref{eq42_1}) for agent $i$. Note that if $b_i \ne 0$, then the agent $i$ has the information of the leader.}
   \label{fig2}
\end{figure*}

\begin{remark}
The unknown term $N_2$ includes disturbance, its time-derivative, the leader's velocity, and its time-derivative. Following the similar procedure as in \cite{yang2015robust,xian2004continuous}, considering the smooth functions for $f$ and $g$ in system dynamics as well as Assumptions \ref{ass1} and \ref{ass3}, one has $||N_2|| \le d_M$ and $||\dot{N}_2||\le d_{dM}$, where $d_M$ and $d_{dM}$ are considered to be unknown.
\end{remark}


Take the time-derivative of (\ref{eq12_1}), recall Remark \ref{rem4}, and consider Assumption \ref{ass5}, the boundedness of the activation function and its derivative, $\vec{x}$, $\dot{\vec{x}}$, $Y$, and $\dot{\mathcal{X}}$ over the compact set of $\Omega_\mathcal{X}$, then the boundedness of $||\dot{\bar{\epsilon}}(\mathcal{X},\dot{\mathcal{X}})||$ follows -- namely, $||\dot{\bar{\epsilon}}(\mathcal{X},\dot{\mathcal{X}})|| \le \bar{\epsilon}_{dM}$.
 We consider $\bar{\epsilon}_{M}$ and $\bar{\epsilon}_{dM}$ to be unknown while authors in \cite{dierks2008neural,fischer2013saturated,patre2008asymptotic} considered these bounds to be known. 

For the agent $i$, one has $||\bar{\epsilon}_i|| \le \bar{\epsilon}_{M_i}$,  $||\dot{\bar{\epsilon}}_i(\mathcal{X}_i,\dot{\mathcal{X}_i})|| \le \bar{\epsilon}_{dM_i}$, $||N_{2_i}|| \le d_{M_i}$, and $||\dot{N}_{2_i}||\le d_{dM_i}$. Let us define $\bar{N}_{2_i} \triangleq  N_{2_i} + \bar{\epsilon}_{i}$. Note that values of $\bar{\epsilon}_{M_i}$, $\bar{\epsilon}_{dM_i}$, $d_{M_i}$, and $d_{{dM}_i}$ are considered to be unknown, and they only are used for analytical purposes.

For the agent $i$, consider (\ref{eq_46_N}) and recall the fact that the activation function and its derivative are bounded. Using the Assumption \ref{ass4}, for an arbitrary set of $\Omega_\vec{x}$, the inputs of the NN are bounded and their derivatives $(\mathcal{X}_i,\dot{\mathcal{X}_i})$ are bounded. Therefore, we can conclude that there exists upper bounds for $\tilde{N}_{1i}$ and $\dot{\tilde{N}}_{1i}$, namely $\|\tilde{N}_{1i}\|\le N_{Mi}$ and $\|\dot{\tilde{N}}_{1i}\|\le N_{dMi}$, which are considered to be unknown. This fact also can be found in \cite{yang2015robust,fan2017asymptotic,xian2004continuous,patre2008asymptotic}.

As continuous differentiability is a sufficient condition for being locally Lipschitz \cite[Corollary 4.1.1]{scholtes2012introduction}, \cite[Theorem 8.4]{khalil} can be written as follows:

\begin{lemma}[\hspace{-.1mm}\cite{xian2004continuous}]\label{khalil}
Consider the close-loop system $\dot{\chi} = f(\chi,t)$, where $f: D \times [0,\infty)$ with $D \in \mathbb{R}^{n}$ being a set containing $\chi=0$. Let a continuously differentiable scalar function $V: D \times [0,\infty)$ satisfy the following inequalities
\begin{equation}\label{eq_khalil}
\begin{split}
    W_1(\chi) \le &V(\chi,t) \le W_2(\chi), \\ 
     \dot{V}(\chi,t) \le & - W_3(\chi), \ \ \ \forall t\ge 0, \forall \chi \in D,
\end{split}    
\end{equation}
where $W_1$, $W_2$ are continuous positive definite functions and $W_3$ is a uniformly continuous positive semi-definite function. If (\ref{eq_khalil}) holds, select a closed ball $B(0,R)$ with the radius of $R>0$, and let $\rho < min_{||\chi||=R} W_1(\chi)$. Solutions of $\dot{\chi} =f(\chi,t)$ with $\chi(t_0) \in \{\chi \in D | \ W_2(\chi) \le \rho \}$ are bounded and 
\begin{equation}\label{eq_khalil1}
    W_3(\chi) \rightarrow 0 \ \text{as} \  t\rightarrow \infty.
\end{equation}
\end{lemma}
\subsection{Stability Analysis}
In this section, we prove that the proposed control law (\ref{eq42_1}) guarantees the stability of the leader-following formation control problem for the class of heterogeneous, second-order, uncertain, nonlinear multi-agent systems (\ref{eq1}). Before we give the main theorem of the paper, the following lemma is stated, which will be used to prove the theorem.
\begin{lemma}\label{lem5}
Let an auxiliary function $L_i(t) \in R$ be defined as follows:
\begin{equation}\label{eq47_1}
\begin{split}
    L_i(t)  \triangleq & r_i^T \big{(}\bar{N}_{2_i} - \kappa_{d_i}\text{sgn}(\zeta_i)\big{)} +\dot{\zeta}^T_i \tilde{N}_{1_i} \\ &+ (W^2_m +V_m^2+Z^2_m)\zeta^T_i \text{sgn}(\zeta_i).
\end{split}
\end{equation}
Set $\kappa_{d_i} = d_{M_i} + \bar{\epsilon}_{M_i} +\frac{1}{k_2} \max\{ d_{dM_i}+ \bar{\epsilon}_{dM_i} + N_{dM_i}+(W^2_m+V^2_m+Z^2_m),k_2 N_{M_i} \} $. Then, the following inequality holds
\begin{equation}\label{eq48_1}
\begin{split}
   M_i \triangleq & \  \kappa_{d_i}||\zeta_i(0)||_1- \zeta^T_i(0) (\bar{N}_{2_i}(0)+\tilde{N}_{1_i}(0) ) \\& - \int_0^t L_i(s)ds \ge 0.
\end{split}\end{equation}
\end{lemma}

\begin{proof}
Similar to \cite{fan2017asymptotic}, from (\ref{eq40}), and by integrating (\ref{eq47_1}) one has
{\small{\begin{equation}\label{eq52_N}
\begin{aligned}
   \int_0^t L_i(s) ds = & \int_0^t  \dot{\zeta_i}^T(s) \big{(} \bar{N}_{2_i}(s) + \tilde{N}_{1_i}(s) \big{)} ds \\ &  
   - \kappa_{d_i} \int_0^t  \dot{\zeta_i}^T(s)  \text{sgn}(\zeta_i(s)) \big{)} ds \\& 
   +k_2 \int_0^t  {\zeta}^T_i(s)  \big{(} \bar{N}_{2_i}(s) - \kappa_{d_i} \text{sgn}(\zeta_i(s))   \\ & 
   + \frac{(W^2_m +V_m^2+Z^2_m)}{k_2} \text{sgn}(\zeta_i(s))  \big{)} ds  \\ 
= & \zeta_i^T(s) (\bar{N}_{2_i}(s)+\tilde{N}_{1_i}(s)) \Big|_0^t +   k_2 \int_0^t  {\zeta}^T_i (s) \big{(} \bar{N}_{2_i}(s)  \\&   
- \kappa_{d_i} \text{sgn}(\zeta_i(s)) + \frac{(W^2_m +V_m^2+Z^2_m)}{k_2} \text{sgn}(\zeta_i (s) ) \\ & -\frac{1}{k_2} \dot{\bar{N}}_{2_i}(s) -\frac{1}{k_2}\dot{\tilde{N}}_{1_i}(s) ds\big{)}  -k_{d_i} \|\zeta_i(s)\|_1 \Big|_0^t.
\end{aligned}
\end{equation}}}Using Cauchy–Schwarz inequality, and the definition of $\kappa_{d_i}$, one has
\begin{equation}\label{eq54_N}
\begin{split}
   \int_0^t L_i(s) \le &
 \zeta_i^T(t) (\bar{N}_{2_i}(t)+\tilde{N}_{1_i}(t)) - \zeta_i^T(0) (\bar{N}_{2_i}(0)+\tilde{N}_{1_i}(0)) \\ & 
-k_{d_i} \|\zeta_i(t)\|_1 +k_{d_i} \|\zeta_i(0)\|_1
\\ \le &    \kappa_{d_i}||\zeta_i(0)||_1- \zeta^T_i(0) (\bar{N}_{2_i}(0)+\tilde{N}_{1_i}(0) ).
\end{split}
\end{equation}
Therefore, (\ref{eq48_1}) is valid. This completes the proof.\hfill  $\blacksquare$
\end{proof}

\begin{remark}\label{remark_6}
We propose to set neurons numbers of the first hidden-layer and the second hidden-layer to $m_1=3n+4$ and $m_2=2(3n+4)+2$, for each agent. From \cite[Theorem 2.4]{ismailov2014approximation}, one can find that a three-layer NN with $3n+4$ inputs, $3n+4$ neurons in the first hidden-layer, $2(3n+4)+2$ neurons in the second hidden-layer, and the sigmoid activation function, can approximate $N_{1_i}$ with an arbitrary accuracy. 
\end{remark}

\begin{theorem}\label{th1}
Let the multi-agent system (\ref{eq1}) be modeled by a strongly connected  directed graph. For each agent, select the number of the hidden-layer neurons in (\ref{eq8}) as $m_1=3n+4$, $m_2=2(3n+4)+2$, and the NN input as (\ref{eq44_1}). Under Assumptions \ref{ass1}-\ref{ass5}, choose $k_1>\frac{1}{2}$ and 
\begin{equation}\label{eq54_k2}
    k_2 >   \frac{1 +\bar{\sigma}(P(L+B))}{\underline{\sigma}(P(L+B))  }.
\end{equation}
Set the control input as 
\begin{equation}\label{eq_control}
\begin{split}
    u_i(t) =& -(k_3+k_4) \zeta_i(t)+(k_3+k_4) \zeta(0) \\& - \int_0^t \Big{(}\hat{W}_i^T(s) \sigma_1 \big{(} \hat{Z}_i^T(s) \sigma_2 (\hat{V}_i^T(s) \mathcal{X}_i(s))\big{)}  \\& + k_2(k_3+k_4) \zeta_i(s) + \kappa_i(s) \text{sgn}(\zeta_i(s)) \Big{)}ds,
\end{split}\end{equation}
where $k_3$ and $k_4$ satisfy the following conditions
\begin{equation}\label{eq53_k4}
    k_3 > \frac{k_2}{2\underline{\sigma}(P(L+B))}, \   k_4 > \frac{\bar{\sigma}(P(L+B))}{2}.
\end{equation}
Set the time-varying gain $\kappa_i(t)$ as (\ref{eq_kappa_N}), and the neural networks weights matrices tuning laws as
{\small{\begin{equation}\label{eq_tune0}
\begin{split}
    \dot{\hat{W}}_i & = \alpha_i \Big{(} k_2\big{(} \hat{\sigma}_1 - \hat{\sigma}_1' \hat{Z}_i^T {\hat{\sigma}_2}-  \hat{\sigma}_1' \hat{Z}_i^T  \hat{\sigma}_2' \hat{V}_i^T{\mathcal{X}_i} \big{)} \zeta_i^T -  \|\zeta_i\|_1 {\hat{W}}_i  \Big{)} ,  \\  
\dot{\hat{Z}}_i & =  \beta_i \Big{(} k_2\big{(} \hat{\sigma}_2 -  \hat{\sigma}_2' \hat{V}_i^T \mathcal{X}_i \big{)} \zeta_i^T \hat{W}_i ^T  \hat{\sigma}_1' - \|\zeta_i\|_1 {\hat{Z}}_i  \Big{)} , \\ 
\dot{\hat{V}}_i &=  \gamma_i \big{(} k_2\mathcal{X}_i  \zeta_i^T \hat{W}_i ^T  \hat{\sigma}_1' \hat{Z}_i ^T \hat{\sigma}_2' - \|\zeta_i\|_1 {\hat{V}}_i  \big{)}.
\end{split}\end{equation}}}\noindent
 Then, all the closed-loop signals are bounded, $\zeta \rightarrow 0$ as $t\rightarrow\infty$, and agents achieve the desired formation and maintain the leader's velocity.  
\end{theorem}

\begin{proof}
Let us define the Lyapunov function candidate as
\begin{equation}\label{eq55}
\begin{split}
    V=&  \frac{1}{2}\vec{e}^T\vec{e} +  \frac{1}{4}\zeta^T \Pi \zeta + \frac{1}{2}r^T \mathcal{G}Q r +M +\frac{1}{2}\tilde{\kappa}^T\tilde{\kappa}+ \bar{y},
\end{split}\end{equation}
where semi-positive term $M \triangleq \sum_{i=1}^N M_i$. The variable $M_i$ is defined using Lemma \ref{lem5} and  $\bar{y}$ is given by
\begin{equation}\label{eq56_1}
\begin{split}
\bar{y} \triangleq & \frac{1}{2} tr(\tilde{W}^T \alpha^{-1} \tilde{W})+\frac{1}{2}tr(\tilde{Z}^T \beta^{-1} \tilde{Z}) +\frac{1}{2}tr(\tilde{V}^T \gamma^{-1} \tilde{V})   
\\ = & \frac{1}{2} \text{vec}({\alpha^{\frac{-1}{2}} \tilde{W}})^T\text{vec}({\alpha^{\frac{-1}{2}} \tilde{W}})+\frac{1}{2} \text{vec}({\beta^{\frac{-1}{2}} \tilde{Z}})^T\text{vec}({\beta^{\frac{-1}{2}} \tilde{Z}})\\& +\frac{1}{2} \text{vec}({\gamma^{\frac{-1}{2}} \tilde{V}})^T\text{vec}({\gamma^{\frac{-1}{2}} \tilde{V}}).
\end{split}
\end{equation}

Let us define $\bar{\kappa} = [\kappa_1,...,\kappa_N]^T \in \mathbb{R}^N$, ${\kappa}_d = [{\kappa}_{d_1},...,\hat{\kappa}_{d_N}]^T \in \mathbb{R}^N$ and $\tilde{\kappa}=[\tilde{\kappa}_1,...,\tilde{\kappa}_N] \in \mathbb{R}^N$, where $\tilde{\kappa}_i \triangleq \kappa_i -\kappa_{d_i}$.
Moreover, let $\chi= [\vec{e}^T,\zeta^T,r^T,\sqrt{M},\text{vec}(\alpha^{\frac{-1}{2}}\hat{W})^T,\text{vec}(\beta^{\frac{-1}{2}}\hat{Z})^T$, $\text{vec}(\gamma^{\frac{-1}{2}}\hat{V})^T,\kappa^T]^T $. It can be seen that the following inequalities are valid for (\ref{eq55})
\begin{equation}\label{eq55_1}
  L_b(\chi) \le V  \le U_b (\chi),
\end{equation}
with  
\begin{equation}
     L_b(\chi) =  \underline{\eta} ||\chi||^2, \
     U_b(\chi) =  \bar{\eta} ||\chi||^2,
\end{equation}
where $\underline{\eta} \triangleq \ \text{min} \{ \frac{1}{2} \underline{g}\underline{\sigma}(Q),\frac{1}{2},\frac{1}{4}\underline{\sigma}(\Pi)\}$, $\bar{\eta}\triangleq \text{max} \{\frac{1}{2} \bar{g}\bar{\sigma}(Q) ,1,\frac{1}{4}\bar{\sigma}(\Pi) \}$. To use Lemma \ref{khalil}, we should study the existence of the Filippov's solution for $\dot{\chi} = f(\chi,t)$, as $\dot{M}$ and $\dot{\kappa}$ are differential
equations with discontinuous right-hand side \cite{de2019formation}. Using differential inclusion \cite{filippov2013differential,aubin2012differential}, an absolutely continuous Filipov's solution exists for 
$\dot{\chi} \in {K} {[\mathcal{F}}] ({\chi},t)$
where an upper semi-continuous, compact and convex set-valued map ${K} {[\mathcal{F}}] ({\chi},t) $ is defined as follows:
\begin{equation}\label{eq56}
{K} {[\mathcal{F}}] ({\chi},t) \triangleq \underset{R>0}{\cap} \underset{\mu (H)=0}{\cap} \overline{co} f(B(\chi,R) \backslash H ,t),
\end{equation}
where $B(\chi,R)$ is the closed ball with center $\chi$ and the radius $R$. Using (\ref{eq56}) and from \cite[Theorem 2.2]{shevitz1994lyapunov}, the time-derivative of Lyapunov function candidate exists almost everywhere:
\begin{equation}\label{eq57}
\dot{V}(\chi)   \stackrel{a.e.}{\in} \underset{\Lambda \in\partial V(\chi) }{\cap} \Lambda^T [{K}^T {[\mathcal{F}}] ({\chi}) ]^T.
\end{equation}
As the Lyapunov function candidate (\ref{eq55}) is continuously smooth, (\ref{eq57}) can be written as follows:
\begin{equation}\label{eq60}
{\small{\begin{aligned}
\dot{V}(\chi)   \stackrel{a.e.}{\in}   \triangledown  V \ {K}[& \dot{\vec{e}}^T,\dot{\zeta}^T, \dot{r}^T, \dot{\mathcal{G}}^T, \dot{M}, \text{vec}{(\alpha^{-\frac{1}{2}}\dot{\hat{W}})}^T, \\ & \text{vec}{(\beta^{-\frac{1}{2}}\dot{\hat{Z}})}^T, \text{vec}{(\gamma^{-\frac{1}{2}}\dot{\hat{V}})}^T, \dot{\tilde{\kappa}}^T ]^T,
\end{aligned}}}\end{equation}
where $\triangledown V \triangleq [\vec{e}^T,\frac{1}{2}\zeta^T \Pi,r^T\mathcal{G}Q, \frac{1}{2}\frac{\partial(r^T\mathcal{G}Qr)}{ \partial \mathcal{G}}, 1, \text{vec}(\alpha^{-\frac{1}{2}}\hat{W})^T,$ $\text{vec}(\beta^{-\frac{1}{2}}\hat{Z})^T,\text{vec}(\gamma^{-\frac{1}{2}}\hat{V})^T,\tilde{\kappa}^T]$. 
Using (\ref{eq56_1}) and (\ref{eq60}), the time-derivative of (\ref{eq55}) is given by
\begin{equation}\label{eq62}
\begin{split}
    \dot{V}  \stackrel{a.e.}{\in}    &  \vec{e}^T\dot{\vec{e}} + \frac{1}{2} \zeta^T \Pi \dot{\zeta} + \frac{1}{2}r^T \dot{\mathcal{G}}Q r+ r^T {\mathcal{G}}Q \dot{r} + \dot{M} +\tilde{\kappa}^T\dot{\tilde{\kappa}} \\ & 
    + tr(\tilde{W}^T \alpha^{-1} \dot{\tilde{W}})+tr(\tilde{Z}^T \beta^{-1} \dot{\tilde{Z}})+tr(\tilde{V}^T \gamma^{-1} \dot{\tilde{V}}).  
\end{split}\end{equation}
Substituting (\ref{eq33}), (\ref{eq40}), (\ref{eq43_1}), and (\ref{eq48_1}) in (\ref{eq62}) yields
\begin{equation}\label{eq63}
{\small{\begin{aligned}
    \dot{V}  =  &  \vec{e}^T (\zeta- k_1 \vec{e}) +  \frac{1}{2} \zeta^T \Pi (r-k_2 \zeta) +  r^T \Big{(}- k_2\zeta  -(k_3+k_4) {r} \\&  
    - \kappa(t)\text{Sgn}(\zeta)  + \tilde{W} ^T \big{(} \hat{\sigma}_1 - \hat{\sigma}_1' \hat{Z}^T {\hat{\sigma}_2}- \hat{\sigma}_1' \hat{Z}^T  \hat{\sigma}_2' \hat{V}^T{\mathcal{X}} \big{)}  \\  
& + \hat{W} ^T  \hat{\sigma}_1' \tilde{Z} ^T \big{(} \hat{\sigma}_2 -  \hat{\sigma}_2' \hat{V}^T \mathcal{X} \big{)} +\hat{W} ^T \big{(}  \hat{\sigma}_1' \hat{Z} ^T \hat{\sigma}_2' \tilde{V}^T \mathcal{X} \big{)} \\ 
 & 
+N_2 + \bar{\epsilon}_M \Big{)}   +\tilde{\kappa}^T\dot{\tilde{\kappa}} -\sum_{i=1}^N \bigg{(}r_i^T \big{(}\bar{N}_{2_i} - \kappa_{d_i}\text{Sgn}(\zeta_i)\big{)} \\ & 
+\dot{\zeta}^T_i \tilde{N}_{1_i} + (W^2_m +V_m^2+Z^2_m) \zeta^T_i \text{Sgn}(\zeta_i) \bigg{)} \\ & 
    + tr(\tilde{W}^T \alpha^{-1} \dot{\tilde{W}})+tr(\tilde{Z}^T \beta^{-1} \dot{\tilde{Z}})+tr(\tilde{V}^T \gamma^{-1} \dot{\tilde{V}}),  
\end{aligned}}}\end{equation}
where by abusing notation, $\text{Sgn}(\zeta_i)=[\text{Sgn}(\zeta_{i1}), ..., \text{Sgn}(\zeta_{in})]^T$. The definition of $\text{Sgn}(x_i)$, for $x_i \in R$ is given as in \cite{aubin2012differential} 
\begin{equation}
\text{Sgn}(x_i) =
\begin{cases}
-1 & \text{if } x_i < 0,\\
[-1, 1] & \text{if } x_i = 0,\\
1 & \text{if } x_i > 0.
\end{cases}
\end{equation}
Reorganizing (\ref{eq63}), from (\ref{eq40}) and (\ref{eq_46_N}), one has
\begin{equation}\label{eq64}
{\small{\begin{aligned}
    \dot{V}  \le  &  \vec{e}^T (\zeta- k_1 \vec{e}) + \frac{1}{2} \zeta^T \Pi (r-k_2 \zeta) +  r^T \big{(}- k_2\zeta \\&  
    -(k_3+k_4) {r}\big{)}  +\tilde{\kappa}^T\dot{\tilde{\kappa}} + \sum_{i=1}^N (\kappa_{d_i}- \kappa_i(t))r_i^T \text{Sgn}(\zeta_i) \\ & 
       - (W^2_m +V_m^2+Z^2_m) \zeta^T_i \text{Sgn}(\zeta_i)  \\ &  
    +\sum_{i=1}^N\Big{(} tr  (\tilde{W}_i^T   \big{(} \alpha_i^{-1}\dot{\tilde{W}}_i+ k_2 \big{(}\hat{\sigma}_1  - \hat{\sigma}_1' \hat{Z}_i^T {\hat{\sigma}_2}\\  
    &- \hat{\sigma}_1' \hat{Z}_i^T  \hat{\sigma}_2' \hat{V}_i^T{\mathcal{X}_i} \big{)} \zeta_i^T)\big{)}+  tr\big{(}\tilde{Z}_i ^T (\beta^{-1} \dot{\tilde{Z}}_i \\ &  
    + k_2 \big{(} \hat{\sigma}_2 -  \hat{\sigma}_2' \hat{V}_i^T \mathcal{X}_i \big{)} \zeta_i^T\hat{W}_i ^T  \hat{\sigma}_1' ) \big{)}   \\&  
    +tr\big{(}\tilde{V}_i^T (\gamma^{-1}_i \dot{\tilde{V}}_i+  k_2\mathcal{X}_i \zeta_i^T\hat{W}_i ^T  \hat{\sigma}_1' \hat{Z}_i ^T \hat{\sigma}_2' ) \Big{)}.  
\end{aligned}}}\end{equation}
By substituting (\ref{eq_kappa}) and (\ref{eq_tune0}) in (\ref{eq64}), one can get
\begin{equation}\label{eq65_old}
{\small{\begin{aligned}
    \dot{V}  \le  &  \vec{e}^T (\zeta- k_1 \vec{e}) + \frac{1}{2}  \zeta^T \Pi (r-k_2 \zeta) +  r^T (- k_2\zeta  -(k_3+k_4) {r} ) \\&  
  + \sum_{i=1}^N \bigg{(} \|\zeta_i\|_1 \Big{(} tr(\tilde{W}_i^T \hat{W}_i)  + tr(\tilde{Z}_i^T \hat{Z}_i) + tr(\tilde{V}_i^T \hat{V}_i) \Big{)} \\ &
  - (W^2_m +V_m^2+Z^2_m) \zeta^T_i \text{Sgn}(\zeta_i) \bigg{)} . 
\end{aligned}}}\end{equation}
Applying Young’s inequality for $tr(\tilde{W}_i^T \hat{W}_i)$, one has $tr(\tilde{W}_i^T \hat{W}_i)\le \frac{W_m^2}{2} - \frac{\|\hat{W}_i\|_F^2}{2}$. Following the same procedure for $tr(\tilde{Z}_i^T \hat{Z}_i)$ and $tr(\tilde{V}_i^T \hat{V}_i)$, one can write
\begin{equation}\label{eq65}
{\small{\begin{aligned}
    \dot{V}  \le  &  \vec{e}^T (\zeta- k_1 \vec{e}) + \frac{1}{2}  \zeta^T \Pi (r-k_2 \zeta) +  r^T (- k_2\zeta  -(k_3+k_4) {r} ) \\&  
  + \sum_{i=1}^N \big{(}\tilde{\kappa}_i^T\dot{\kappa}_i  - \tilde{\kappa}^T_{i} r^T_i  \text{Sgn}(\zeta_i) \big{)} . 
\end{aligned}}}\end{equation}
From (\ref{eq_2}) and applying Cauchy–Schwarz inequality and Young’s inequality, the following result can be obtained 
\begin{equation}\label{eq_cauchy}
\begin{split}
    -k_2 r^T  \zeta \le & \frac{k_2}{2} \underline{\sigma}(P(L+B))||\zeta||^2 + k_2 \frac{||r||^2}{2\underline{\sigma}(P(L+B))} .
\end{split}\end{equation}
Substituting (\ref{eq_cauchy}) in (\ref{eq65}) and from (\ref{eq_2}) one can write 
\begin{equation}\label{eq66}
\begin{split}
    \dot{V}  \le  &   - (k_1-\frac{1}{2}) ||\vec{e}||^2 - \big{(}\frac{k_2}{2} \underline{\sigma}(P(L+B)) \\ & 
    - \frac{1 +\bar{\sigma}(P(L+B))}{2} \big{)} ||\zeta||^2     -(k_3+k_4  \\ & -\frac{\bar{\sigma}(P(L+B))}{2} - \frac{k_2}{2\underline{\sigma}(P(L+B))} ) ||r||^2 ). 
\end{split}\end{equation}
From the given conditions for gains $k_1,k_2,k_3,k_4$ in the Theorem \ref{th1}, one has
\begin{equation}\label{eq67}
\begin{split}
    \dot{V}  \le  &   - \lambda_1 ||\vec{e}||^2  - \lambda_2 ||\zeta||^2     - \lambda_3 ||r||^2 , 
\end{split}\end{equation}
where $\lambda_1$, $\lambda_2$, and $\lambda_3$ are positive constants. Consequently, (\ref{eq67}) can be rewritten as 
\begin{equation}\label{eq70}
\begin{split}
    \dot{V}  \le  &   -U(\chi),
\end{split}\end{equation}
where $ U(\chi) \triangleq - \underline{\lambda} ||[\vec{e}^T,\zeta^T,r^T]^T ||^2$, with  $\underline{\lambda} = \text{min} (\lambda_1,\lambda_2,\lambda_3)$. Note that in (\ref{eq70}),  $U(\chi)$ is a positive semi-definite function over $\Omega_\mathcal{X}$. From (\ref{eq55}) and (\ref{eq70}), one has $V \in L_\infty$ over $\Omega_\mathcal{X}$; hence, $\vec{e},\zeta,r,\tilde{W},\tilde{Z},\tilde{V}$ and $\tilde{\kappa}$ are bounded on set $\Omega_\mathcal{X}$. Consequently, from (\ref{eq33}) $\delta \in L_\infty$. From Assumption \ref{ass1} and (\ref{eq31}), one can conclude that $\vec{x} \in L_\infty$. From (\ref{eq40}), as $r \in L_\infty$ and $\zeta \in L_\infty$, one can obtain $\dot{\zeta} \in L_\infty$. Using Lemma \ref{lem_l+b}, Assumptions \ref{ass1}-\ref{ass3}, and (\ref{eq34}), it can be deduced that $u \in L_\infty$. Based on Assumption \ref{ass4}, (\ref{eq42_1}) and (\ref{eq43_1}) are also bounded. From the boundedness of $\vec{e},\delta, \zeta,r$ and the closed-loop terms,  $U(\chi)$ is uniformly continuous, where Lemma \ref{khalil} can be applied. Let $\Omega_\mathcal{X}=\{\chi |\ ||\chi|| \in B(0,R)\}$, and the initial conditions belong to the compact set  $\mathcal{S} \subset \Omega_\mathcal{X}$ as
\begin{equation}
    \mathcal{S}\triangleq \{\chi \in \Omega_\mathcal{X} | U_b(\chi) \le \underline{\eta} R^2  \}, 
\end{equation}
where the origin is within the set of $\mathcal{S}$. Then, the Lemma \ref{khalil} is used to conclude $||\vec{e}||^2\rightarrow0, ||\zeta||^2 \rightarrow0, ||r||^2 \rightarrow0, \text{as} \  t\rightarrow \infty, \ \forall \chi(0) \in  \mathcal{S}$. From (\ref{eq31}) and (\ref{eq33}), one can conclude that as $||\vec{e}|| \rightarrow 0$, and  $||\delta||\rightarrow 0$, (\ref{eq6}) holds. Consequently, by increasing $R$, the semi-global asymptotic stability is obtained \cite{xian2004continuous}. This completes the proof. \hfill $\blacksquare$
\end{proof}


\section{Simulation Results}
In order to study the performance of the proposed method, we consider a multi-agent system with five two-link robot arms. The directed graph which models the system is shown in Fig. \ref{fig3}. The  dynamics of each agent is given by \cite{IET}
\begin{equation}\label{ex4_}
\begin{array}{lr}
\dot{p}_{i1}=p_{i2}, \\
M_i(p_{i1})\dot{p}_{i2}+V_i(p_{i1},p_{i2})p_{i2}+G_i(p_{i1})+w_{i}=\tau_i,
\end{array}
\end{equation}
for $i\in\{1,...,5\}$, where $p_{i1}=[p_{i11},p_{i12}]^T \in \mathbb{R}^2$ and $p_{i2}=[p_{i21},p_{i22}]^T\in \mathbb{R}^2$ are the joint position and velocity states of $i$-th robot arm, respectively. In (\ref{ex4_}), the control input is $\tau_i\triangleq u_i$. Expressions for the inertia matrix $M_i(p_{i1})$, the Coriolis matrix  $V_i(p_{i1},p_{i2})$, and the gravitational vector $G_i(p_{i1})$ in (\ref{ex4_}) can be found in \cite{IEEE_TAES,IET}.
In (\ref{ex4_}), $w_i$ represents a bounded disturbance in the dynamics of each two-link robot arm. The parameters of the each robot arm are $g=9.8m/s^2$, $r_{i1}=1m$, $r_{i1}=1m$, $m_{i1}=0.8kg$, and $m_{i2}=1.7kg$, $w_i=[-0.12cos(t),0.1sin(t)]^T $, $i\in\{1,...,5\}$. The leader dynamics is given by
\begin{equation}\label{leader}
\begin{array}{lr}
\dot{p}_l=v_l, \\
\dot{v}_{l}= \begin{bmatrix}
-p_{l_1}+0.2(1-p_{l_1}^2)v_{l_1} \\
-p_{l_2}+0.3(1-p_{l_2}^2)v_{l_2}
\end{bmatrix},
\end{array}
\end{equation}
where $p_{l}=[p_{l_1},p_{l_2}]^T \in \mathbb{R}^2$ and $v_{l}=[v_{l_1},v_{l_2}]^T \in \mathbb{R}^2$ are the position and velocity states of the leader. The initial conditions for the five robots are given as: $p_{11}(0)=[2.1,0]^T$, $p_{21}(0)=[0,2.5]^T$,  $p_{31}(0)=[-1.1,2]^T$, $p_{41}(0)=[-1.8,0.7]^T$, $p_{51}(0)=[-1,-1.7]^T$, $p_{i2}(0)=[0,0]^T, \forall i \in\{1,...,5\}$. The desired displacement with respect to the leader for each agent is given as: $d_1= [0,d]^T$, $d_2= [-d sin(\frac{2\pi}{5}),d cos(\frac{2\pi}{5})]^T$, $d_3= [-d sin(\frac{\pi}{5}),-d cos(\frac{\pi}{5})]^T$, $d_4= [ d sin(\frac{\pi}{5}),-d cos(\frac{\pi}{5})]^T$ and $d_5= [d sin(\frac{2\pi}{5}),d cos(\frac{2\pi}{5})]^T$ with $d=1$.
The leader initial states are $p_l(0)=[1,-1]^T$ and $v_l(0)=[0,0]^T$, respectively. All weights of edges in Fig. \ref{fig3} are equal to $one$, as well as the leader's edge weight ($b_1=1$). Therefore, the adjacency matrix and graph Laplacian are given by
\begin{equation}\label{78}
\small{A=
\begin{bmatrix}
0&0&0&0&1   \\
1&0&0&0&0\\
1&1&0&0&0\\
0&0&1&0&1\\
0&0&1&0&0\\
\end{bmatrix},
L=
\begin{bmatrix}
1&0&0&0&-1   \\
-1&1&0&0&0\\
-1&-1&2&0&0\\
0&0&-1&2&-1\\
0&0&-1&0&1\\
\end{bmatrix}.}
\end{equation}

We selected $10$ neurons for the input layer, $10$ neurons for the first hidden layer, $22$ neurons for the second hidden layer, and $2$ neurons for the output layer for each agent. The activation function for these layers is a sigmoid function. The activation function for the output layer is linear. The constants $\alpha_i$, $\beta_i$, $\gamma_i$, $\forall i$, are chosen as $\alpha_i= \frac{1}{20}$, $\beta_i= \frac{1}{20}$, $\gamma_i= \frac{1}{20}$. The gains were selected as: $k_1=4$, $k_2=37.5$, $k_3=380$, and $k_4=2$.

Figs. \ref{fig4}-\ref{fig10} show the simulation results for the multi-agent system (\ref{ex4_}). In Fig. \ref{fig4}, the performance of the multi-agent system is shown where dash-dotted lines indicate the multi-agent system achieving its desired formation. Fig. \ref{fig5} illustrates trajectories of agents and the leader. The velocities of agents, as well as the leader, are shown as Fig. \ref{fig6}, where velocities of agents become identical to the leader velocity. The Frobenius norm of the NN weights matrices of $\hat{V}_i$, $\hat{Z}_i$ and $\hat{W}_i$ are displayed as Fig. \ref{fig7}. Fig. \ref{fig8} demonstrates the time-varying gains $\kappa_i$. To avoid that $\kappa_i(t)$ reaches high values, a dead-zone with size of $b=0.005$, has been implemented, similar to \cite{yang2015robust,fan2017asymptotic}. The error signals ${e}_i(t)$, their time-derivatives $\delta_i(t)$, and the filtered errors $\zeta_i(t)$ are shown in Fig. \ref{fig9}. The control input of each agent is shown in Fig. \ref{fig10}.

These results verify the performance of the proposed method based on RISE feedback and the NN-based controller. To compare our proposed method with some other existing works, we considered results on leader-following consensus control of the nonlinear system with unknown nonlinearity \cite{lewis,IET}. We included the desired displacements and applied these methods similar to \cite[p.~127]{mesbahi2010graph}. We define the average of the control inputs cost function (ACI), $\nu(t)$ and average of the formation errors cost function (AFE), $\vartheta(t)$, as follows:
\begin{equation}
\begin{split}
    \nu(t) =& \frac{1}{2N} \sum_{i=1}^N u_i^T(t) u_i(t), \\
    \vartheta(t) =& \frac{1}{2N} \sum_{i=1}^N ||{e}_i(t)||_1.
\end{split}\end{equation}

We use these two functions as performance indices to evaluate and compare the effectiveness of the proposed methods in comparison with \cite{lewis,IET}. Fig. \ref{fig11}(a), shows the semi-logarithmic graph of ACI function. 
However, the proposed method (solid-line) energy consumption is almost the same as the other two methods. Fig. \ref{fig11}(b) indicates that our proposed method settled faster and converged to zeros in contrast with two other methods where the error remains bounded.
\begin{figure}[t]
    \centering
    \includegraphics[scale=1]{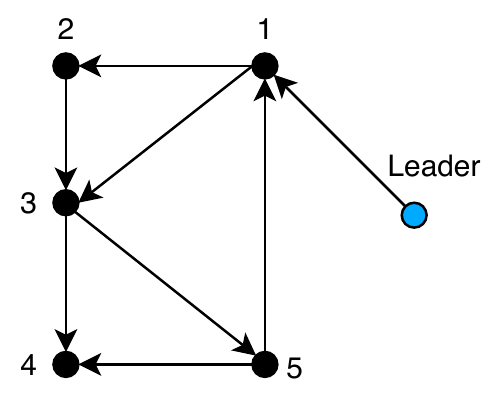}    
  \caption{The directed graph of five agents with the leader as the root of the spanning tree.}
   \label{fig3}
\end{figure}  
\begin{figure}[t]
    \centering
    \includegraphics[scale=.45]{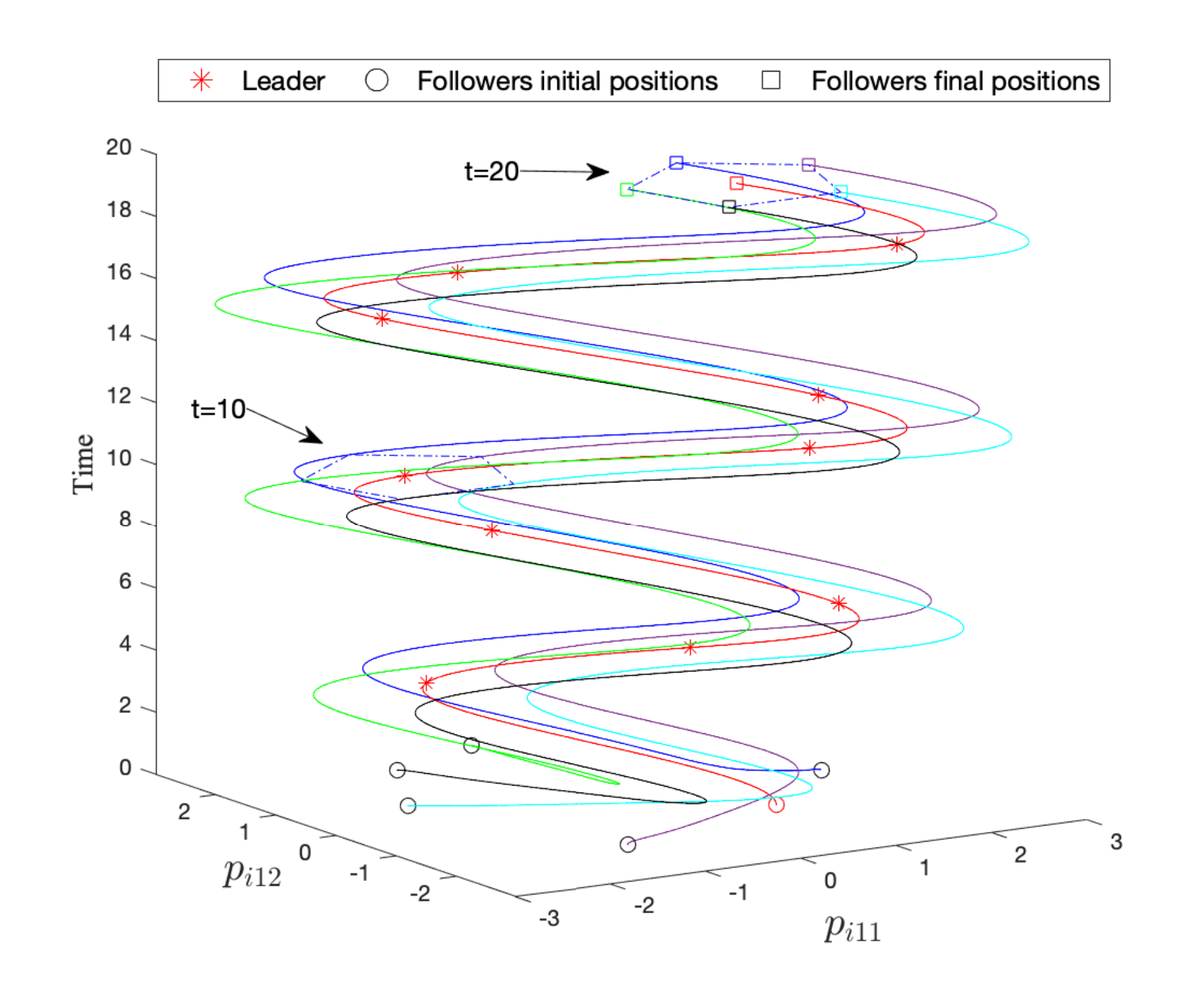}    
  \caption{Performance of the multi-agent system (\ref{ex4_}).}
   \label{fig4}
\end{figure}  
\begin{figure}[th]
    \centering
    \includegraphics[scale=.45]{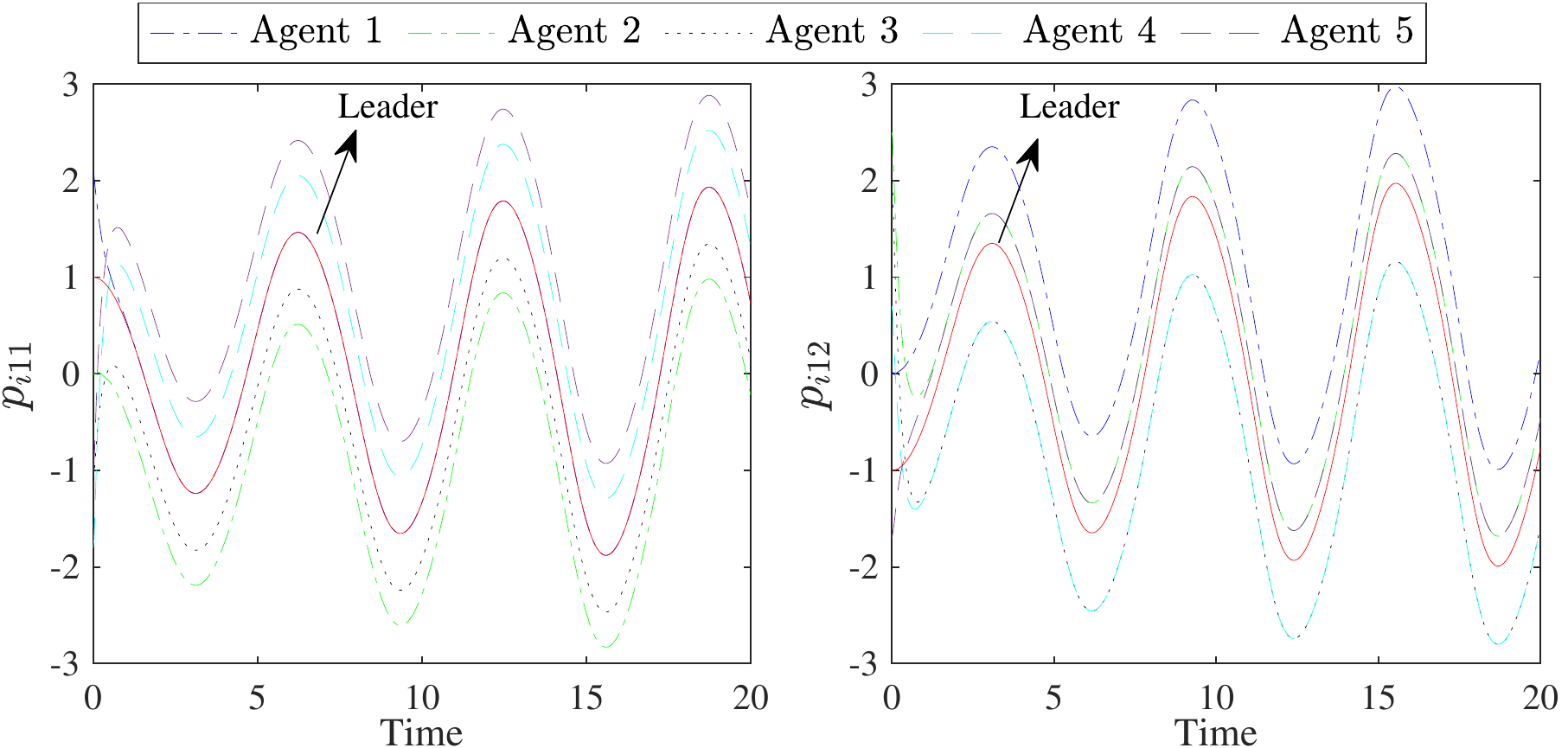}    
  \caption{Evolution of agents and the leader position states (leader is the solid line).}
   \label{fig5}
\end{figure}  
\begin{figure}[th]
    \centering
    \includegraphics[scale=.45]{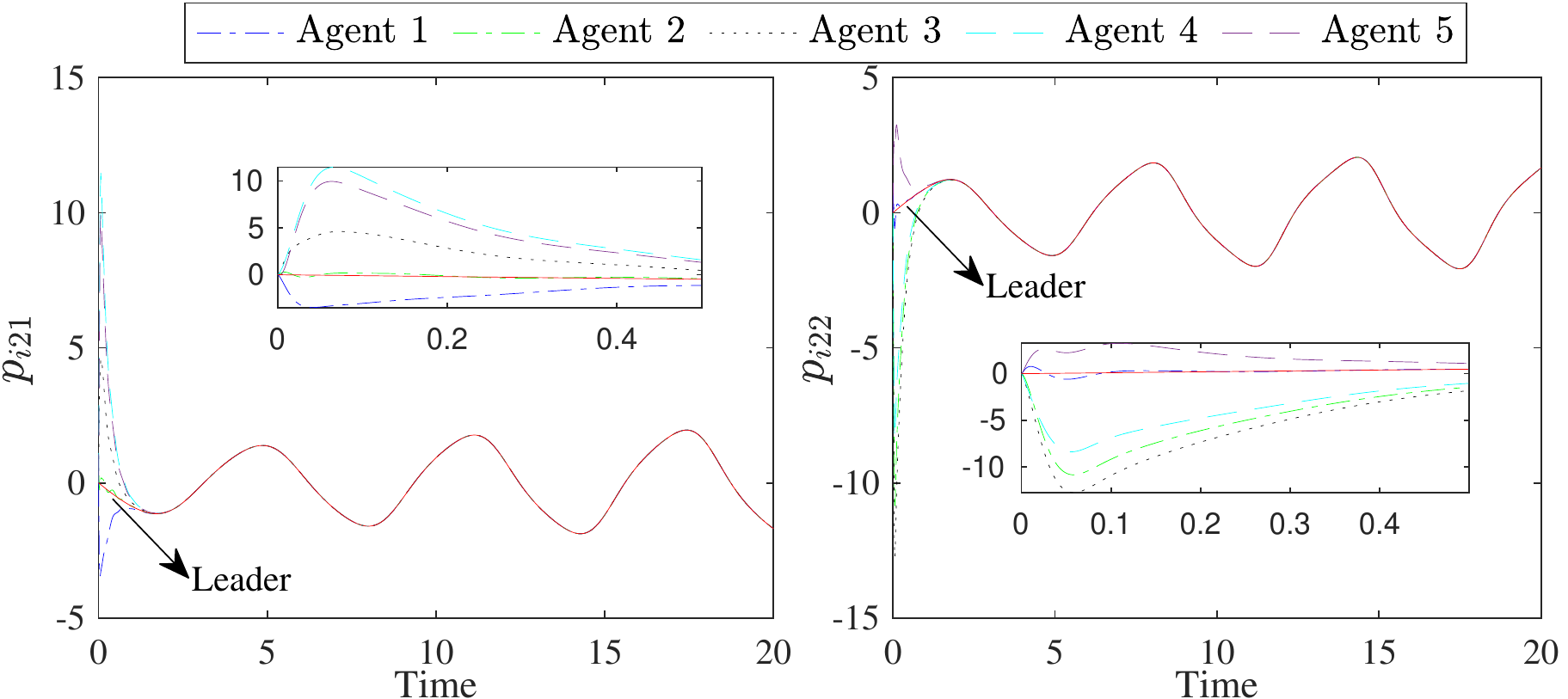}    
  \caption{Evolution of agents and the leader velocity states (leader is the solid line).}
   \label{fig6}
\end{figure}  

\begin{figure}[th]
    \centering
    \includegraphics[scale=.45]{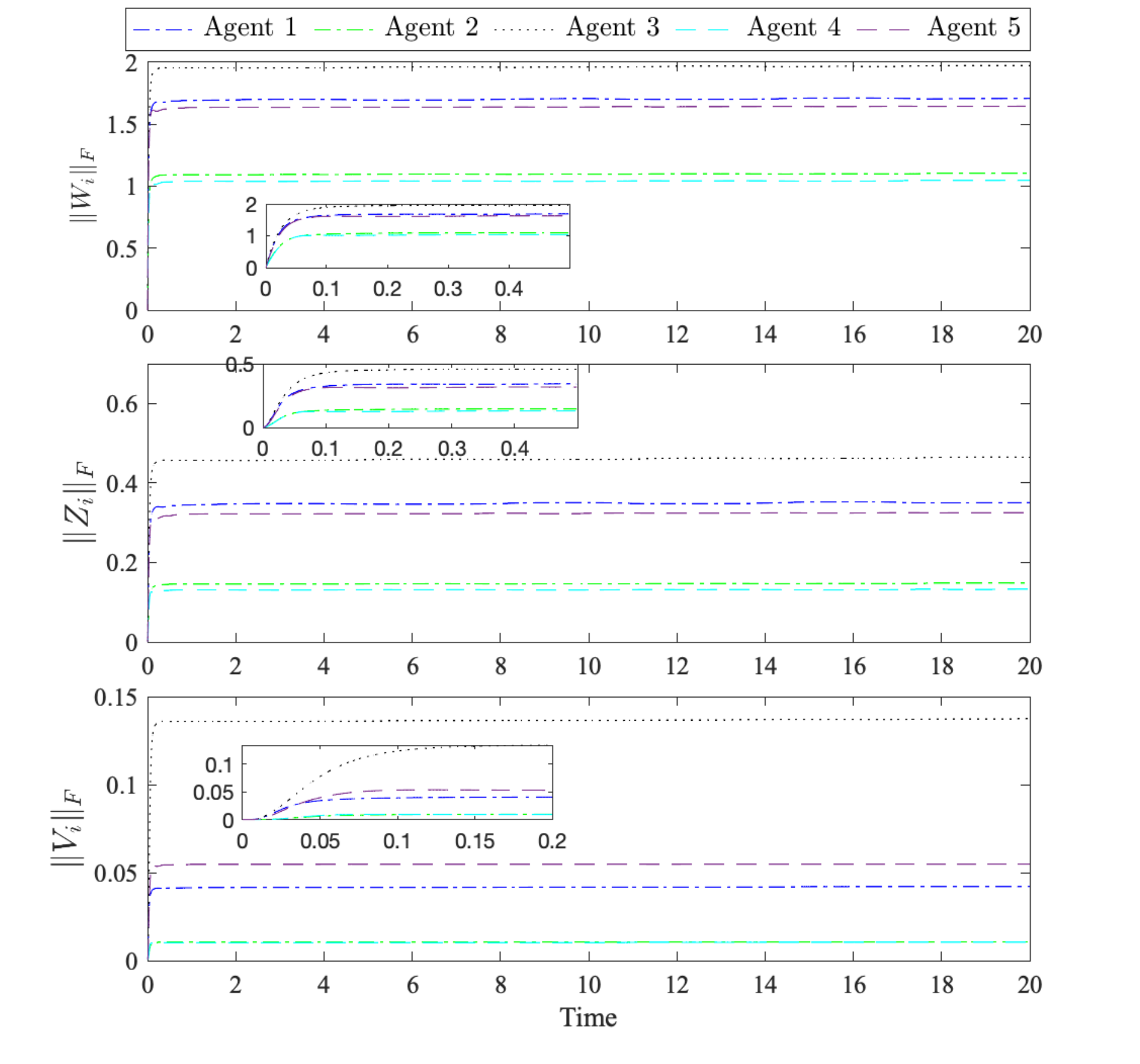}
  \caption{Frobenius Norm of NN weights matrices for agents.}
   \label{fig7}
\end{figure}  

\begin{figure}[th]
    \centering
    \includegraphics[scale=.45]{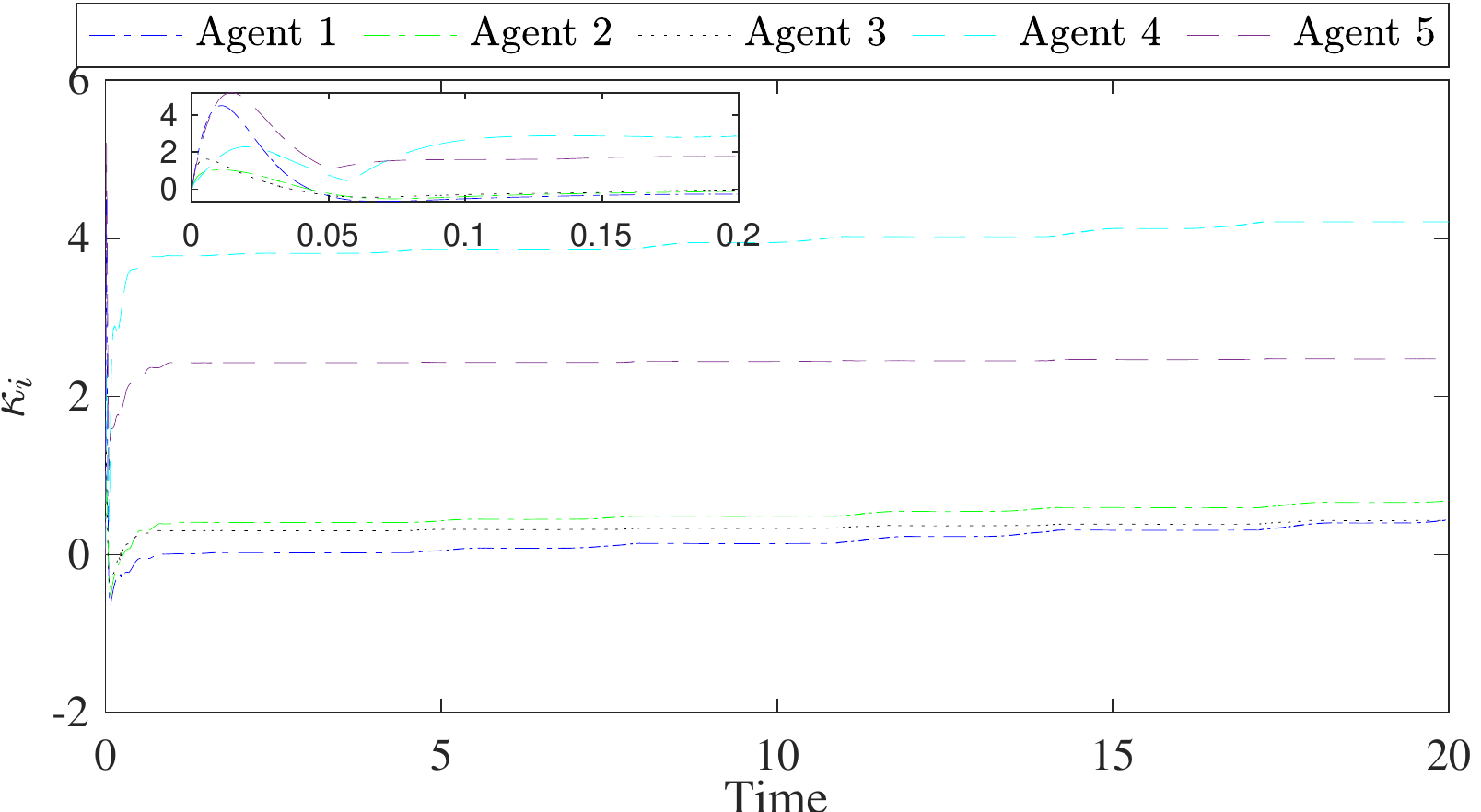}
  \caption{Time-variant gain $\kappa_i$ of each agent.}
   \label{fig8}
\end{figure}

\begin{figure}[th]
    \centering
    \includegraphics[scale=.45]{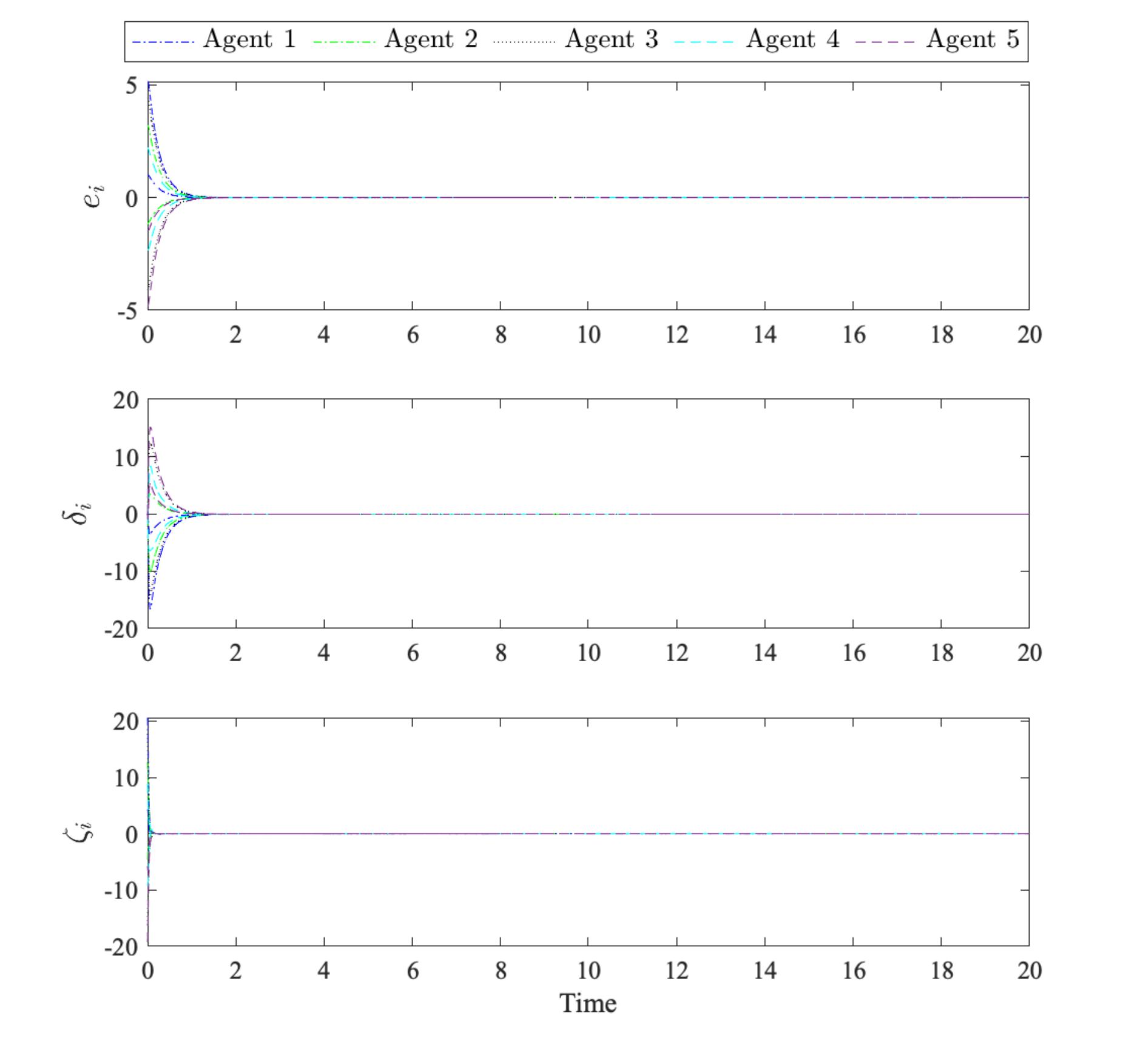}
  \caption{Error signals of $\vec{e}$, $\delta$, and $\zeta$ for each agent.}
   \label{fig9}
\end{figure}

\begin{figure}[th]
    \centering
    \includegraphics[scale=.45]{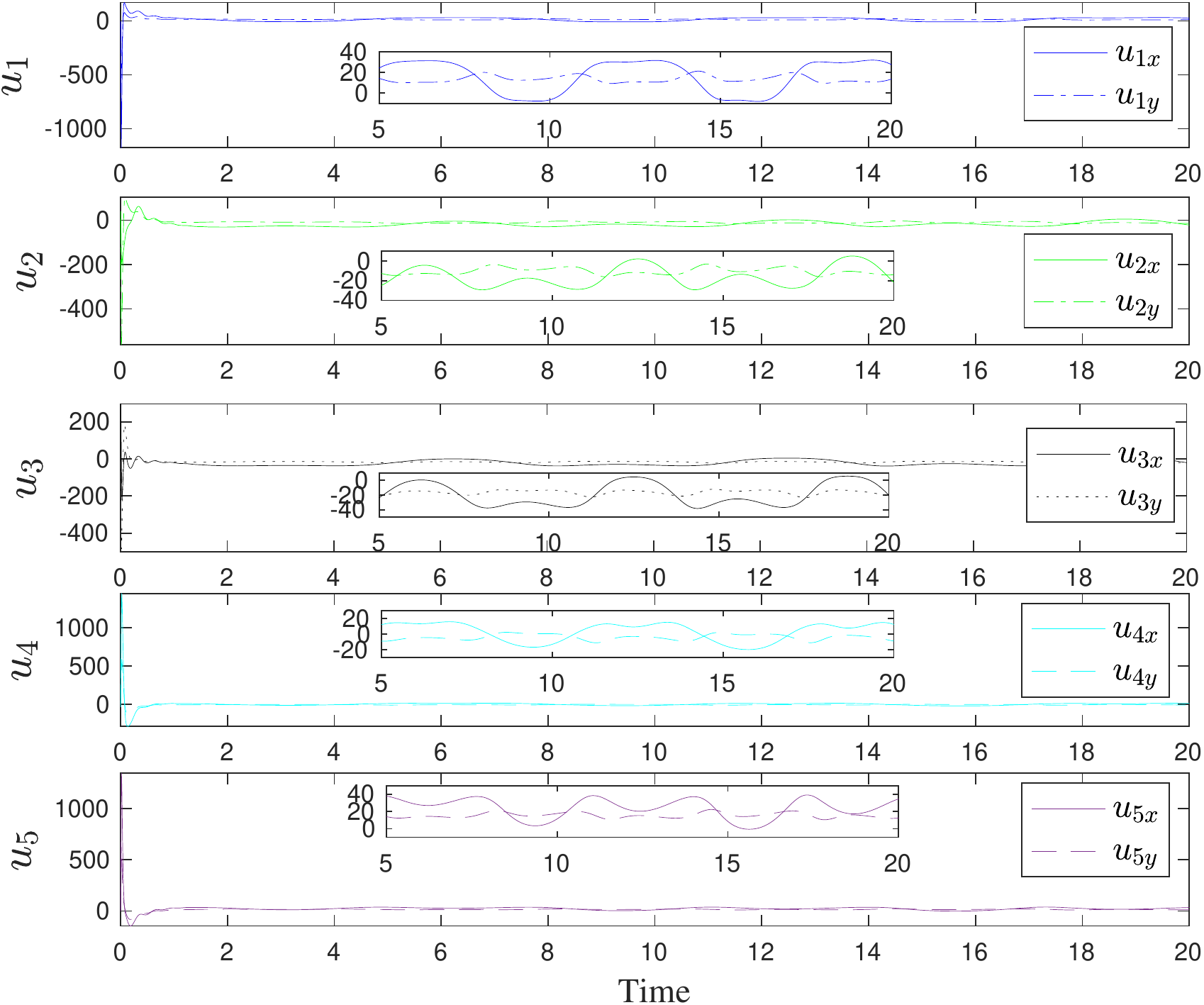}
  \caption{Control inputs of each agent.}
   \label{fig10}
\end{figure}

\begin{figure}[th]
    \centering
    \includegraphics[scale=.65]{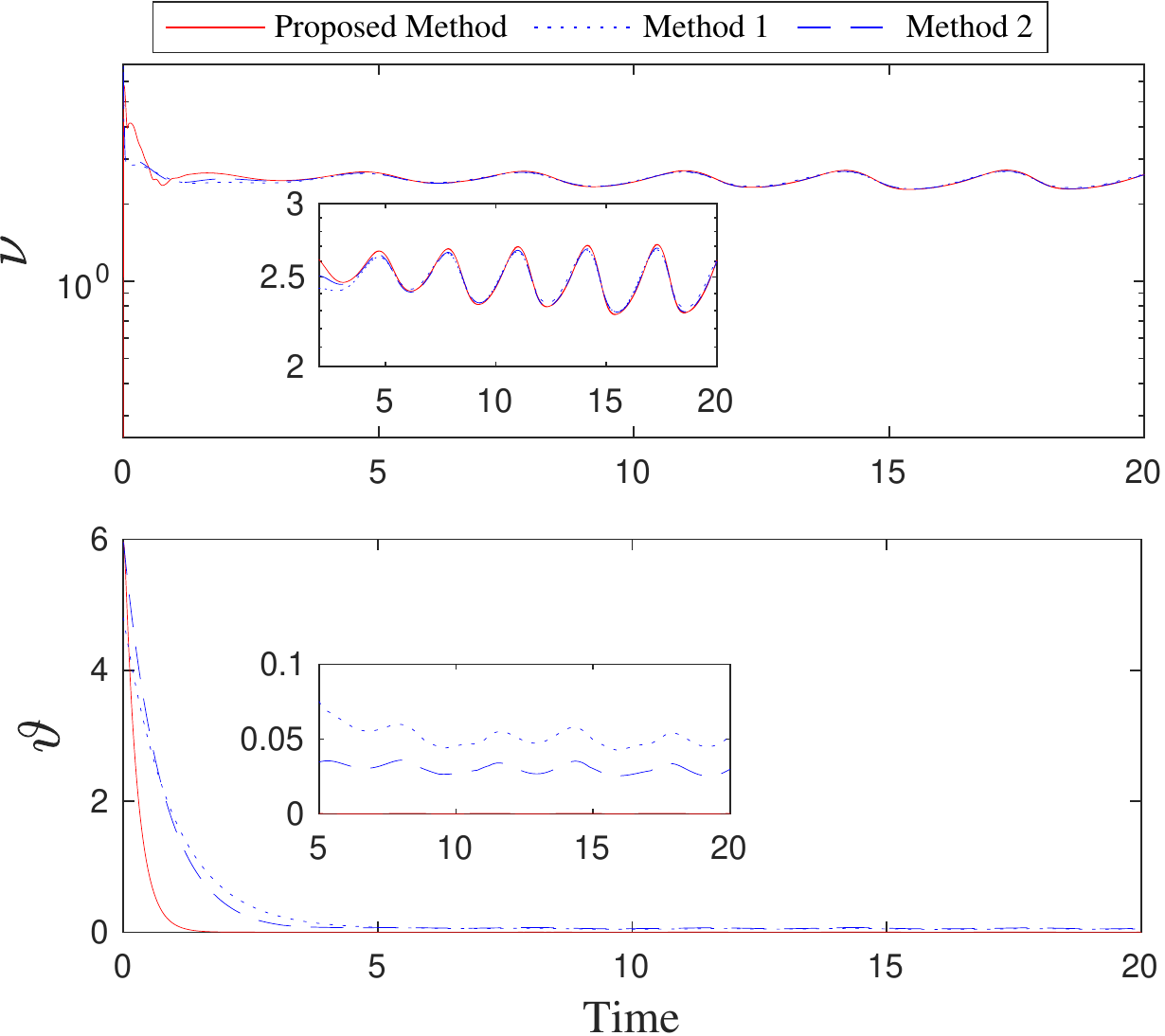}
  \caption{(a) ACI function comparison between proposed method and Method 1 \cite{IET} and Method 2 \cite{lewis} ($y$-axis is in logarithmic scale); (b) AFE function comparison between proposed method and Method 1 \cite{IET} and Method 2 \cite{lewis}.}
   \label{fig11}
\end{figure}


\section{Conclusion}
This paper developed the leader-following formation control of the heterogeneous, second-order, uncertain, input-affine, nonlinear multi-agent systems modeled by a directed graph. The unknown nonlinearity in the dynamics of a multi-agent system was approximated by a tunable, three-layer NN consisting of an input layer, two hidden layers, and an output layer. The proposed method can \textit{a priori} set the number of neurons in each layer of NN. The NN weights tuning laws were derived using the Lyapunov theory. A robust integral of the sign of the error feedback control with an NN was developed to guarantee semi-global asymptotic leader tracking. The boundedness of the close-loop signals and asymptotic leader tracking formation were proven using the Lyapunov stability theory. The paper results were compared with two previous results, which showed the effectiveness of the proposed method.

\bibliographystyle{IEEEtran}
\bibliography{bibfile}

\end{document}